\documentclass[acmsmall]{acmart}
\renewcommand\footnotetextcopyrightpermission[1]{} 
\usepackage{fancyhdr}
\RequirePackage{tgtermes}
\RequirePackage{newtxtext}
\RequirePackage{bm}
\RequirePackage{endnotes}

\usepackage{algorithm}
\usepackage{algpseudocode}
\usepackage{tikz}

\usepackage[caption=false]{subfig}
\usepackage{dsfont}
\usepackage{enumitem}
\usepackage{tikz}
\usetikzlibrary{positioning, arrows.meta, shapes}
\usepackage{wrapfig}

\usetikzlibrary{calc}                        
\definecolor{jobgray}{RGB}{150,150,150}
\definecolor{latepurple}{RGB}{150,110,222}
\definecolor{earlypink}{RGB}{242,96,205}
\tikzset{
  jb/.style={rectangle, fill=jobgray, draw=black, line width=0.5pt, text=white,
             font=\small, minimum width=0.55cm, inner sep=0pt},
  server/.style={circle, draw=black, line width=1pt, minimum size=1.25cm},
}
\newcommand{\gseg}[5]{
  \fill[jobgray] (#1,#3) rectangle (#2,#4);
  \draw[line width=0.4pt] (#1,#3) rectangle (#2,#4);
  \node[text=white,font=\footnotesize] at ($(#1,#3)!.5!(#2,#4)$) {#5};}
\newcommand{\gframe}[2]{
  \draw[-{Latex[length=1.6mm]},line width=0.8pt] (0,0)--(#1,0);
  \draw[-{Latex[length=1.6mm]},line width=0.8pt] (0,0)--(0,3.0);
  \node[font=\small] at (#1*0.42,3.0) {#2};
  \useasboundingbox (-1,-0.7) rectangle (#1+0.5,3.4);}
\newlength{\rowh}
\newcommand{\colcell}[1]{\resizebox{!}{\rowh}{#1}}
\newcommand{\colstack}[2]{
  \begin{tabular}{@{}c@{}}\colcell{#1}\\[2.5pt]\colcell{#2}\end{tabular}}
\newcommand{\gbox}[5]{
  \fill[#5] (#1,#3) rectangle (#2,#4);
  \draw[line width=0.4pt] (#1,#3) rectangle (#2,#4);}
\newcommand{\gtick}[2]{\draw[line width=0.6pt] (#1,-0.15)--(#1,0.15);
  \node[below,font=\footnotesize] at (#1,-0.18) {#2};}

\usepackage{hyperref}
\usepackage[capitalise]{cleveref}

\newcommand{\ind}[1]{\mathcal{I}_{\{#1\}}}
\newcommand{\E}[1]{\mathbb{E}[#1]}
\renewcommand{\P}[1]{\mathbb{P}(#1)}
\newcommand{\rb}{\boldsymbol{r}}
\newcommand{\bb}{\boldsymbol{b}}
\newcommand{\bbb}{\bb\bb}
\newcommand{\bad}{\textsc{Bad}}
\newcommand{\good}{\textsc{Good}}
\newcommand{\neutral}{\textsc{Neutral}}
\mathchardef\hy="2D

\newcommand{\new}[1]{#1}

\usepackage{natbib}

\begin{document}




\title{Outperforming Multiserver SRPT at All Loads}



\author{Izzy Grosof}
\email{izzy.grosof@northwestern.edu}
\author{Daniela Hurtado-Lange}
\email{daniela.hurtado@kellogg.northwestern.edu}
\affiliation{%
  \institution{Northwestern University}
  \city{Evanston}
  \state{Illinois}
  \country{USA}
}

\makeatletter
\newtheorem*{rep@theorem}{\rep@title}
\newcommand{\newreptheorem}[2]{%
\newenvironment{rep#1}[1]{%
 \def\rep@title{#2 \ref{##1}}%
 \begin{rep@theorem}}%
 {\end{rep@theorem}}}
\makeatother

\newreptheorem{theorem}{Theorem}
\newreptheorem{lemma}{Lemma}

\begin{abstract}%
A well-designed scheduling policy can unlock significant performance improvements with no additional resources. Multiserver SRPT (SRPT-k) is known to achieve asymptotically optimal mean response time in the heavy-traffic limit, as load approaches capacity. No better policy is known for the M/G/k queue in any regime. We introduce a new policy, SRPT-Except-k+1 \& Modified SRPT (SEK-SMOD), which is the first policy to provably achieve lower mean response time than SRPT-k. SEK-SMOD outperforms SRPT-k across all loads and all job size distributions. 
The key idea behind SEK-SMOD is to prioritize large jobs over small jobs in specific scenarios to improve server utilization, and thereby improve the response time of subsequent jobs in expectation. Our proof is a novel application of hybrid worst-case and stochastic techniques to relative analysis, where we analyze the deviations of our proposed SEK-SMOD policy away from the SRPT-k baseline policy. Furthermore, we design Practical SEK (a simplified variant of SEK-SMOD) and empirically verify the improvement over SRPT-k via simulation.
\end{abstract}%





\maketitle

\section{Introduction}

The question of reducing wait times with limited resources is ubiquitous in operations and engineering. In particular, strategic scheduling is known to be an effective way to substantially reduce the average time a job spends in the system (and, consequently, the customer waiting time) without increasing the number of servers. This effectiveness has been observed in various practical applications.

\cite{bray2016courthouse} study delays in cases handled by a courthouse in Italy. Specifically, they design a case-level scheduling policy in which judges estimate the number of hearings needed for each case and schedule all of them in advance, prioritizing \new{hearings} for cases nearing completion. Compared to the old scheduling policy, where judges scheduled each hearing after completing the previous one, the courthouse reduced average case duration by approximately 12\% without increasing resources or the number of judges. Similar scheduling policies have been implemented in multiple courthouses afterward, with comparable results \citep{azaria2023TOC,azaria2024jerusalem-courthouse}.

In computer systems, designing algorithms that minimize response time without adding more resources is essential, as the number of CPUs and GPUs available is typically costly to increase. For example, one can think of a single machine with a fixed number of cores to process tasks. Moreover, this machine accurately estimates the sizes of arriving tasks. Hence, it can use this information to schedule strategically. For the rest of this paper, we mainly consider this example as an immediate application. \citet{Kul-Ful2011-example-scheduling} and \citet{Zhu-et-al2012-example-scheduling} describe some additional examples on the importance of scheduling in computer systems with limited resources.

One of the most popular scheduling policies in the literature is First-Come First-Served (FCFS), where the jobs are processed in order of arrival. This policy is simple to implement and is often considered fair. Due to its simplicity, its analysis is tractable, and has been studied for decades ~\new{\citep{loulou1973multi,kollerstrom1974heavy-1,kollerstrom1979heavy-ii,gupta2010inapproximability,gupta2011markov,grosof-et-al2022wcfs}}. However, FCFS does not utilize servers strategically: both long and short jobs receive the same priority. Hence, it does not perform well in terms of response time. If one prioritizes shorter jobs over longer ones, the mean response time decreases significantly. Specifically, one can \new{prioritize} jobs according to their remaining amount of work, and preemptively process those jobs with the Shortest-Remaining-Processing-Time (SRPT) policy. For decades, SRPT has been the gold standard for scheduling jobs in single-server queues, as it minimizes the mean response time \citep{schrage1968proof-srpt-opt}. However, optimal scheduling for multiserver systems is much less understood.

\new{A natural generalization of SRPT to the multiserver setting is SRPT-$k$, where the $k$ jobs with the least remaining processing time are scheduled. If there are fewer than $k$ jobs available, some of the servers are idle. This policy has been extensively studied in both the worst-case and stochastic scheduling literature,
with early works including \cite{leonardi1997approximating,leonardi2007SRPT-worst-case1},
\cite{down2006multi}\footnote{\new{Note that the proof of Theorem 2.1 in \cite{down2006multi} is flawed. See \cite[Appendix E]{grosof-etal-2018-srpt-k}.}},
and an approximate version of SRPT-$k$ in \cite{harchol2005multi}.
\cite{grosof-etal-2018-srpt-k} proved the first theoretical optimality result for SRPT-$k$,
showing that it minimizes mean response time in the heavy-traffic limit. Empirically, SRPT-$k$ results in low mean response times across all loads, not just heavy traffic, but no theoretical results outside heavy traffic are known.}

\new{An} open question in the literature \new{for over a decade has been} the existence of a scheduling policy that outperforms SRPT-$k$ \new{\cite[Section 24.3]{harchol2013performance}}. In this paper, we answer positively with our policy SRPT-Except-$k+1$ \& Modified SRPT (SEK-SMOD). The contributions of this paper are:
\begin{enumerate}[topsep=0pt]
    \item We devise and introduce the SEK-SMOD policy, which is defined in reference to a coupled SRPT-$k$ system, in \cref{sec:def-sek-smod}.
    \item We prove that SEK-SMOD achieves a smaller mean response time compared to SRPT-$k$ under all loads, all job size distributions, and all numbers of servers $k \ge 2$ (\cref{thm:main}).
    \item We develop a proof method based on a strategic combination of stochastic analysis with \new{worst-case analysis (equivalently, all-sample-path analysis)}
    as well as relative analysis. We discuss the proof philosophy in \cref{sec:overcome}, the proof outline in \cref{sec:proof-outline}, and give the proof in \cref{sec:stochastic} (stochastic analysis) and \cref{sec:worst-case} (worst case analysis).
    \item We propose the Practical SEK policy (a simplified variant of the SEK-SMOD policy) and empirically demonstrate that it outperforms SRPT-$k$, via simulation, in \cref{sec:empirical}.
    \item \new{We discuss practical considerations for the implementation and use of SEK in \cref{sec:discuss}.}
\end{enumerate}
\subsection{Challenges in beating SRPT-k}\label{sec:intro-challenges}

It is well known in the literature that SRPT minimizes the mean response time in single-server queues with known job sizes \citep{schrage1968proof-srpt-opt}, for any load.   Further, the multiserver SRPT-$k$ policy has recently been proven to minimize the mean response time in heavy traffic \citep{grosof-etal-2018-srpt-k}.
SRPT-$k$ also minimizes total response time for all jobs in a no-arrivals (batch) setting, \new{as proven by} \cite{graham1979optimization}, combining prior results by \cite{conway1967miller} \new{and} \cite{mcnaughton1959scheduling}.
As a consequence, a scheduling policy that performs better than SRPT-$k$ needs to focus on moderate load.

An important pathway for improving upon SRPT-$k$ is rearranging the order in which jobs are served to improve future parallelism.
SRPT-$k$ finishes the smallest jobs first, leaving only a few large jobs. As a result, it reaches states where it cannot fully utilize its servers sooner than other policies. \new{We present a concrete example in \cref{fig:intro-example}}. 
Taking advantage of future parallelism requires a \emph{global} view of the state of the system: Is the system close to empty, with just a few jobs left, or are plenty of jobs present?

In contrast, existing multiserver scheduling analysis has overwhelmingly focused on \emph{index} policies, where each job is evaluated in isolation to produce an index, and the $k$ jobs of best index are served. Existing analyses for multiserver index policies include FCFS-$k$ \citep{loulou1973multi,kollerstrom1974heavy-1,kollerstrom1979heavy-ii,gupta2010inapproximability,gupta2011markov,grosof-et-al2022wcfs}, class-based priority policies \citep{sleptchenko-et-al2005exact,harchol2005multi}, SRPT-$k$ and related index policies \citep{grosof-etal-2018-srpt-k}, M-Gittins-$k$ and M-SERPT-$k$ \citep{scully-et-al2021optimal}, and Gittins-$k$ \citep{scully-et-al2020gittins}.
Index policies cannot take a global view of the system,
and so cannot look ahead to improve future parallelism.
We need a new method of multiserver analysis which can handle non-index scheduling policies.

However, deviating from SRPT-$k$ to improve future parallelism \new{brings} a risk of degrading the mean response time. Intuitively, if jobs arrive in the immediate future, then SRPT-$k$ suffers no loss of parallelism, and any rearrangement is detrimental.
On the other hand, if no jobs arrive before the system empties, we are in the batch scenario described above, \new{where SRPT-$k$ is optimal}. 
Only arrivals in the intermediate future can lead to better outcomes for another policy.
No prior analysis has attempted to balance these probabilities
to ensure a consistent improvement in expectation.

A scheduling policy that performs better than SRPT-$k$ needs to focus on moderate load, needs to examine the global state of the system, and needs to weigh the probabilities of arrivals at different intervals in the future.
The SEK-SMOD policy fulfills all of these goals,
making possible the first improvement in mean response time over SRPT-$k$.

\subsection{Overcoming challenges}
\label{sec:overcome}

We devise the SRPT-Except-$k+1$ \& Modified SRPT (SEK-SMOD) scheduling policy, which we prove improves upon the mean response time of the SRPT-$k$ policy, for every job size \new{distribution $F_S$}, arrival rate $\lambda$, and number of servers $k \ge 2$. Our result is formally stated as \cref{thm:main}.

SEK-SMOD improves upon SRPT-$k$ by identifying a narrow set of system states where future parallelism can be improved by deviating away from the SRPT-$k$ policy. Each such state involves $k+1$ jobs present in the system,
and in each case SEK-SMOD serves the job of largest remaining size in the system, leaving the second-largest job in the queue.
Otherwise, SEK-SMOD serves the same jobs as SRPT-$k$.
These deviations form the SEK part of the combined SEK-SMOD policy.

We prove that each deviation is beneficial, in expectation.
The effect on long-term mean response time of this deviation is neutral if no jobs arrive for a long duration after the deviation, and beneficial under the most common arrival patterns, given the Poisson arrival process.
The deviation is only harmful under a limited set of scenarios involving arrivals soon after the deviation point, when the predicted improvements in future parallelism fail to materialize.
We characterize the frequency and benefit of these deviations
through stochastic analysis of the SEK policy in \cref{sec:stochastic}.

\new{As an example of an advantageous deviation point and these three arrival scenarios, consider \cref{fig:intro-example}.
At time $t=0$, there are three jobs in the system and their remaining sizes are 0.3, 0.4 and 10.
It is most reliably advantageous to deviate from SRPT-$k$ when the $k$ shortest-remaining-size jobs are very short.
As shown in 
Fig.~\hyperref[fig:ex-intro-queues]{\ref{fig:intro-example}a}, 
SRPT-$k$ schedules the two shortest jobs, and SEK deviates from SRPT-$k$ and schedules the shortest and largest jobs. If there are no arrivals in the near future, the servers'  busy periods follow the timelines represented in 
Fig.~\hyperref[fig:ex-intro-no-arrivals]{\ref{fig:intro-example}b}. 
Note that SEK empties the system at $t=20$ (since each server processes at rate 1/2), while SRPT-$k$ empties the system at $t=20.6$. Further, SRPT-$k$ uses only one of the two servers for 19.8 time units, while SEK only does so for 18.6 time units.
Therefore, if there are no arrivals in the near future, SEK is more efficient. 
Specifically, if the next arrival occurs at a time $t\geq 1.4$ as in 
Fig.~\hyperref[fig:ex-intro-late-arrivals]{\ref{fig:intro-example}c}
, the SEK system is in a better state when that arrival happens.
Hence, in this context, SEK performs better than SRPT-$k$.
This scenario is common because the cutoff time of $1.4$ is small, thanks to our choice to deviate when both short jobs are nearing completion.
The only harmful scenario for SEK is if the next arrival is very close to the deviation point. In our example, this occurs if the next arrival happens at any time $t<1.4$, as shown in 
Fig.~\hyperref[fig:ex-intro-early-arrivals]{\ref{fig:intro-example}d}. 
In this case, the SRPT-$k$ system benefits from having one server available for a longer time. This scenario is unlikely because the arrivals follow a Poisson process and due to our choice of the deviation point. Also, the harm is small compared to the benefits of the other two scenarios, as we show in our proof.

}

\begin{figure}
\mbox{}
\hspace{-20pt}
{\subfloat[Queue diagram.\label{fig:ex-intro-queues}]{\colstack
{\begin{tikzpicture}[server/.append style={minimum size=1.05cm}]
  \node[font=\small] at (1.7,1.12) {SRPT-$k$};
  \draw[line width=1pt] (0,0.66)--(3.2,0.66)--(3.2,-0.66)--(0,-0.66);
  \node[jb, minimum height=1.15cm] at (2.7,0) {10};
  \node[server] (s1) at (4.0,0.58) {}; \node[server] (s2) at (4.0,-0.58) {};
  \node[jb, minimum height=0.50cm] at (s1) {0.3};
  \node[jb, minimum height=0.60cm] at (s2) {0.4};
  \useasboundingbox (-0.2,-1.2) rectangle (4.7,1.45);
\end{tikzpicture}}
{\begin{tikzpicture}[server/.append style={minimum size=1.05cm}]
  \node[font=\small] at (1.7,1.12) {SEK};
  \draw[line width=1pt] (0,0.66)--(3.2,0.66)--(3.2,-0.66)--(0,-0.66);
  \node[jb, minimum height=0.60cm] at (2.7,0) {0.4};
  \node[server] (t1) at (4.0,0.58) {}; \node[server] (t2) at (4.0,-0.58) {};
  \node[jb, minimum height=0.50cm] at (t1) {0.3};
  \node[jb, minimum height=1.15cm] at (t2) {10};
  \useasboundingbox (-0.2,-1.2) rectangle (4.7,1.45);
\end{tikzpicture}}}}
\hspace{-12pt}
{\subfloat[Timeline if no arrivals.\label{fig:ex-intro-no-arrivals}]{\colstack
{\begin{tikzpicture}[x=0.32cm,y=0.5cm]
  \gframe{14}{SRPT-$k$}
  \gseg{0}{1.5}{1.6}{2.4}{0.3}\gseg{1.5}{11.5}{1.6}{2.4}{10}\gseg{0}{2.2}{0.5}{1.3}{0.4}
  \gtick{0}{0}\gtick{1.5}{.6}\gtick{2.2}{.8}\gtick{11.5}{20.6}
\end{tikzpicture}}
{\begin{tikzpicture}[x=0.32cm,y=0.5cm]
  \gframe{14}{SEK}
  \gseg{0}{1.5}{1.6}{2.4}{0.3}\gseg{1.5}{3.7}{1.6}{2.4}{0.4}\gseg{0}{10}{0.5}{1.3}{10}
  \gtick{0}{0}\gtick{1.5}{.6}\gtick{3.7}{1.4}\gtick{10}{20}
\end{tikzpicture}}}}
\hspace{-15pt}
{\subfloat[Timeline if late arrivals.\label{fig:ex-intro-late-arrivals}]{\colstack
{\begin{tikzpicture}[x=0.32cm,y=0.5cm]
  \gframe{21}{SRPT-$k$}
  \gseg{0}{1.5}{1.6}{2.4}{0.3}\gseg{1.5}{11.5}{1.6}{2.4}{10}\gseg{0}{2.2}{0.5}{1.3}{0.4}
  \gbox{11.5}{18.8}{1.6}{2.4}{latepurple}\gbox{10}{17.3}{0.5}{1.3}{latepurple}
  \draw[latepurple,line width=1.5pt,dash pattern=on 3pt off 2pt] (10,0)--(10,2.7);
  \gtick{0}{0}\gtick{1.5}{.6}\gtick{2.2}{.8}\gtick{11.5}{20.6}
\end{tikzpicture}}
{\begin{tikzpicture}[x=0.32cm,y=0.5cm]
  \gframe{21}{SEK}
  \gseg{0}{1.5}{1.6}{2.4}{0.3}\gseg{1.5}{3.7}{1.6}{2.4}{0.4}\gseg{0}{10}{0.5}{1.3}{10}
  \gbox{10}{17.3}{1.6}{2.4}{latepurple}\gbox{10}{17.3}{0.5}{1.3}{latepurple}
  \draw[latepurple,line width=1.5pt,dash pattern=on 3pt off 2pt] (10,0)--(10,2.7);
  \gtick{0}{0}\gtick{1.5}{.6}\gtick{3.7}{1.4}\gtick{10}{20}
\end{tikzpicture}}}}
\hspace{-15pt}
{\subfloat[Timeline if early arrivals.\label{fig:ex-intro-early-arrivals}]{\colstack
{\begin{tikzpicture}[x=0.32cm,y=0.5cm]
  \gframe{14}{SRPT-$k$}
  \gseg{0}{1.5}{1.6}{2.4}{0.3}\gseg{1.5}{11.5}{1.6}{2.4}{10}\gseg{0}{2.2}{0.5}{1.3}{0.4}
  \gbox{2.2}{9.9}{0.5}{1.3}{earlypink}
  \draw[earlypink,line width=1.5pt,dash pattern=on 3pt off 2pt] (2,0)--(2,2.7);
  \gtick{0}{0}\gtick{1.5}{.6}\gtick{2.2}{.8}\gtick{11.5}{20.6}
\end{tikzpicture}}
{\begin{tikzpicture}[x=0.32cm,y=0.5cm]
  \gframe{14}{SEK}
  \gseg{0}{1.5}{1.6}{2.4}{0.3}\gseg{1.5}{3.7}{1.6}{2.4}{0.4}\gseg{0}{10}{0.5}{1.3}{10}
  \gbox{3.7}{11}{1.6}{2.4}{earlypink}
  \draw[earlypink,line width=1.5pt,dash pattern=on 3pt off 2pt] (2,0)--(2,2.7);
  \gtick{0}{0}\gtick{1.5}{.6}\gtick{3.7}{1.4}\gtick{10}{20}
\end{tikzpicture}}}}
\hspace{-20pt}
\mbox{}

\caption{Example comparing the SRPT-$k$ and SEK policies depending on the arrivals timeline. Comparison of the SRPT-$k$ and SEK scheduling policies depending on the arrival times of future jobs. SEK is beneficial in (c) under late arrivals, and SRPT-$k$ performs better in (d) under early arrivals, shortly after the deviation point.}
\label{fig:intro-example}
\end{figure}

In order to limit the downside of these rare but harmful arrival sequences, the policy switches to the SMOD part of SEK-SMOD.
\new{We define SMOD by reference to a coupled SRPT-$k$ system, if the SEK divergence had never been performed.
SMOD's goal is to reach a state at least as preferable as the SRPT-$k$ system's state as soon as possible.
To that end, SMOD pairs up its jobs with those of the coupled SRPT-$k$ system in least-remaining-size order, and prioritizes jobs with remaining size greater than or equal to the corresponding SRPT-$k$ job, with a least-remaining-size fallback prioritization.}

As we show in \cref{sec:worst-case},
SMOD has a variety of helpful properties relative to a coupled SRPT-$k$ system with the same arrival sequence, on all sample paths.
For example, the total work present in the SMOD system is always at most as much as the work present in the SRPT-$k$ system,
under any \new{sample path of arrivals}.
These worst-case results allow us to bound the consequences of the rare but harmful arrival patterns, and prove that the common arrival patterns always have more impact in expectation.
\new{Throughout the paper, when we refer to ``worst-case'' results, we mean results which hold on all sample paths.}
We thereby complete the proof that SEK-SMOD always improves upon SRPT-$k$.

Our analysis of the SEK-SMOD policy combines two recently developed scheduling analysis frameworks: \emph{hybrid} worst-case and stochastic analysis, and \emph{relative analysis}, where a novel policy is analyzed in terms of its infrequent deviations away from a well-studied baseline policy.
%
%

%
Combining the power and generality of hybrid worst-case and stochastic arguments with the divergence-by-divergence framework of relative analysis allows us to prove that SEK-SMOD outperforms SRPT-$k$ for all job size distributions and all arrival rates.

\section{Prior work}

We discuss prior work on SRPT scheduling in \cref{sec:prior-srpt},
on hybrid stochastic \& worst-case analysis in \cref{sec:prior-hybrid},
on relative analysis in \cref{sec:prior-relative},
and
on general multiserver scheduling analysis in \cref{sec:prior-multiserver}. 
We have conducted preliminary work on the design and empirical performance of our proposed policy; however, we omit the details for anonymity.


\subsection{Shortest Remaining Processing Time (SRPT)}
\label{sec:prior-srpt}

In a single-server setting, the Shortest Remaining Processing Time (SRPT) is well understood. Its mean response time has been exactly characterized in the M/G/1 case \citep{schrage-et-al1966srpt-characterize}, and it has been shown to achieve minimal mean response time in a single-server queue under an arbitrary arrival sequence \citep{schrage1968proof-srpt-opt}.

In multiserver systems, SRPT was first studied in a batch setting, where all jobs are initially present, and no jobs arrive over time. \citet{mcnaughton1959scheduling} proved that preemption is not useful for minimizing total response time in the batch setting: any preemptive schedule can be matched by a nonpreemptive schedule.
\citet{conway1967miller} proved that the Shortest Processing Time (SPT) policy (the nonpreemptive equivalent of SRPT) 
is a member of a class of policies that all equally minimize total response time.
\citet[Section~4.4.1]{graham1979optimization} pointed out that these results can be combined to prove that SPT (and SRPT) are optimal in the multiserver batch setting with preemption allowed.

A more recent study of multiserver SRPT (also known as SRPT-$k$) focused on the heavy-traffic limit, where the load approaches the capacity for a fixed number of servers. In this limit, SRPT's mean response time was proven to converge to the optimal possible mean response time, that is, $\E{T^{SRPT-k}}/\E{T^{OPT}} \to 1$ as load $\rho \to 1$, for any number of servers $k \ge 2$ \citep{grosof-etal-2018-srpt-k}. 
In the same paper, the Preemptive Shortest Job First (PSJF-$k$)
and Remaining Size Times Original Size (RS-$k$) policies
were also proven to converge to optimal in the heavy-traffic limit. 
Note that heavy-traffic optimality does not rule out improvement \new{at any or all loads, but merely implies that the multiplicative improvement ratio vanishes in the heavy traffic limit.}

\subsection{Stochastic \& worst-case hybrid analysis}
\label{sec:prior-hybrid}

A key step in the proof of heavy-traffic optimality for SRPT-$k$
is a worst-case coupling argument,
which establishes that the relevant work at some threshold $x$ 
(that is, the total remaining size of all jobs with remaining size no more than $x$) 
present in the SRPT-$k$ system is never more than $kx$ more
than the relevant work in a resource-pooled single-server SRPT system
experiencing the same arrivals \citep[Lemma~2]{grosof-etal-2018-srpt-k}.
\cite{leonardi2007SRPT-worst-case1} had previously proved an equivalent result, where they showed that in the worst case, 
SRPT-$k$ achieves an order-optimal competitive ratio for mean response time with respect to the optimal clairvoyant policy \citep[Lemma~2]{leonardi2007SRPT-worst-case1}.
This worst-case relevant work bound is combined with stochastic tagged-job analysis to upper bound the mean response time of SRPT-$k$.

A further development of the stochastic-worst-case hybrid analysis technique is found in the heavy-traffic optimality proof of the guardrails dispatching policy, when combined with SRPT scheduling after jobs are dispatched to their servers \citep{grosof-etal-2019guardrails}.
Worst-case proofs bound the difference in relevant work between
any two servers at any threshold \citep[Lemma 1]{grosof-etal-2019guardrails},
as well as between the guardrails system and a resource-pooled single-server system \citep[Lemma 2]{grosof-etal-2019guardrails}. \new{These worst-case results resemble and build on worst-case results for the closely related Immediate-Dispatching (IMD) policy of \cite{avrahami2003minimizing}.}
These worst-case arguments are again combined with a stochastic tagged-job analysis to upper bound the mean response time of the guardrails dispatching policy.

Our paper builds upon this hybrid framework: our worst-case results in \cref{sec:worst-case} compare the SMOD portion of the SEK-SMOD policy against the SRPT-$k$ policy, while past works made worst-case comparisons with much simpler resourced-pooled systems.
Our stochastic results in \cref{sec:stochastic} also incorporate tagged-job arguments, as well as the relative analysis framework.

\subsection{Relative analysis}
\label{sec:prior-relative}

Relative analysis involves comparing two scheduling policies:
A baseline policy and a tweaked policy which deviates from that baseline policy in a limited and controlled fashion.
The tweaked policy is then analyzed deviation-by-deviation,
with the goal of proving that the net expected effects of those deviations improve the performance metric of interest.

\new{For example, the Fair Sojourn Protocol (FSP) policy was analyzed in the M/G/1 system, in comparison to the Processor Sharing (PS) policy \citep{friedman2003fairness}. FSP prioritizes jobs in the order in which PS would complete the jobs, resulting in sample-path response time improvement over PS shown via a relative-analysis argument.}

\new{The relative analysis approach was later employed} to analyze the Nudge scheduling policy for the M/G/1 queue, demonstrating that it achieves stochastic dominance across its response time distribution relative to the FCFS policy, and a multiplicative improvement in its asymptotic tail of response time \citep{grosof-etal-nudge-2021}.

This work likewise employs the relative analysis framework, comparing SEK-SMOD to SRPT-$k$ on a deviation-by-deviation level. However, the added complexity of the multiserver setting brings new challenges compared to the M/G/1 setting of the \new{FSP and Nudge analyses}.

\subsection{Multiserver scheduling analysis: index policies}
\label{sec:prior-multiserver}
The vast majority of existing multiserver scheduling analysis focuses on \emph{index} policies, where the scheduling policy assigns each job an index based on the job's individual characteristics, and selects the $k$ jobs with best indices, where $k$ is the number of servers. Index policies do not use global information, like the total number of jobs in the system, making them easier to analyze.

Analysis of multiserver index policies started with First-Come First-Served (FCFS-$k$), where the heavy-traffic mean response time in the M/G/$k$ queue
was proven to resemble that of a resource-pooled M/G/1 system \citep{loulou1973multi,kollerstrom1974heavy-1,kollerstrom1979heavy-ii,grosof-et-al2022wcfs}.
Subsequent work has proven similar but much looser bounds for a much broader class of asymptotic environments, beyond heavy traffic \citep{yuan-et-al2025simple-explicit,hong2025new}.
Exact stationary-distribution characterizations have also been proven under exponential service time assumptions, for FCFS-$k$ and two-class priority scheduling \citep{sleptchenko-et-al2005exact}.

More recently, there has been a family of works on more complex index policies under more general job size distributions, proving that such policies are optimal or near-optimal in the heavy-traffic limit, such as SRPT-$k$ and related policies in the known-size setting \citep{grosof-etal-2018-srpt-k}. In the unknown-size setting, a variety of policies have been explored and heavy-traffic optimality results proven, such as Foreground-Background (Least Attained Service) \citep{grosof-etal-2018-srpt-k}, Monotonic Gittins and Monotonic Shortest Expected Remaining Processing Time \citep{scully-et-al2021optimal},
and the Gittins Index policy \citep{scully-et-al2020gittins,hong-et-al2024gittins-ggk}.

All of this prior work has focused on index policies. By contrast, SEK-SMOD is not an index policy. The SEK portion of the policy only diverges from SRPT-$k$ service when exactly $k+1$ jobs are present in the system, allowing it to prioritize future parallelism more effectively.

\section{Model and definitions}

In this paper, we study the known-size M/G/$k$ queueing system.
Jobs arrive according to a Poisson process with rate $\lambda$,
and have \new{initial size $S$ drawn i.i.d. from a general job size distribution with cdf $F_S$}.
The system has $k \ge 2$ servers, each of which serves jobs at rate $1/k$: a job of size \new{$s$} completes after $k\new{s}$ time in service.
We focus on the case where the load $\rho = \lambda \E{S} < 1$, so that stability is possible.
We consider a general scheduling policy $\pi$,
which selects up to $k$ jobs to serve at any point in time.
Service is preempt-resume, so jobs can be paused and returned to the queue with no loss of work. 

The state of the system at time $t$ is given by a sorted vector $\rb(t)$ of remaining sizes,
where $0 < r_1(t) \le r_2(t) \le \ldots \le r_{n(t)}(t)$ and $n(t)$ is the number of jobs in the system.
\new{ The scheduling policy $\pi$ observes the jobs' remaining size and has access to an arbitrary auxiliary Markovian state, derived from past events. 
For example, the FCFS policy would use the auxiliary state to track the order in which the jobs arrive. Our SEK-SMOD policy likewise uses auxiliary state information, as discussed in \cref{sec:def-sek-smod}. The scheduling policy $\pi$ must be Markovian with respect to this combined state descriptor.}

We call the system positive recurrent if there exists a fixed state, consisting of no jobs present and a fixed auxiliary state, such that the system returns to this fixed state almost surely and in finite mean time.
This paper uses the terms `positive recurrence' and `stability' interchangeably, and focuses on stable systems.

We also consider a coupled system consisting of a pair of M/G/$k$ queues experiencing an identical arrival process under two different scheduling policies.
We write the joint state of the coupled system as a pair of sorted vectors of the same length,
with padding indices added at the front of the vector of the system with \new{fewer} jobs present. \new{Either or both scheduling policies may also carry auxiliary state information. }

Specifically, suppose that the individual states of the two queues are the vectors $\rb^A(t)$
and $\rb^B(t)$, for two policies $A$ and $B$.
We define $N^\pi(t) := |\rb^\pi(t)|$ to be the number of jobs in a given system under policy $\pi\in\{A,B\}$ at time $t$. Then, we define the joint state $\bbb(t):=(\bb^A(t),\bb^B(t))$ as follows, depending on the relation between $N^A(t)$ and $N^B(t)$:

\begin{itemize}
    \item $N^A(t) = N^B(t)$: When the two systems have the same number of jobs, we define $\bb^A(t) = \rb^A(t)$ and $\bb^B(t) = \rb^B(t)$.
    \item $N^A(t) > N^B(t)$: When $A$ has more jobs \new{than} $B$, we define $\bb^A(t) = \rb^A(t)$. Let $d(t) = N^A(t) - N^B(t)$ be the number of additional jobs in system $A$ at time $t$.
    We define $b^B_i(t) = 0$ for all $i \le d(t)$, and define $b^B_i(t) = r^B_{i-d(t)}(t)$ for all $i \ge d(t)+1$.
    \item $N^A(t) < N^B(t)$: When $B$ has more jobs than $A$, we proceed symmetrically. We define $\bb^B(t) = \rb^B(t)$,
    let $d(t) = N^B(t) - N^A(t)$, and define $b^A_i(t) = 0$ for all $i \le d(t)$, and $b^A_i(t) = r^A_{i-d(t)}(t)$ for all $i \ge d(t)+1$.
\end{itemize}

We define a padding index to be an index $i$ for which $b^A_i(t)$ or $b^B_i(t)$ is zero. The number of padding indices \new{is} $d(t)=|N^A(t) - N^B(t)|$. We use $W^\pi(t)$ to denote the total remaining work at time $t$, in the system under policy $\pi$. Specifically, $W^\pi(t) = \sum_i r_i^\pi(t)=\sum_i b_i^\pi(t)$.

For example, for the M/G/2 systems under policies $A$ and $B$ \new{shown} in \cref{fig:notation-example}, the state descriptors are \new{given} in \cref{tab:notation-example}.

\begin{figure}
{\subfloat[Example under policy $A$.\label{fig:notation-ex-A}]{%
\resizebox{0.25\linewidth}{!}{\begin{tikzpicture}[server/.append style={minimum size=1.0cm}]
  \node[font=\large] at (2.0,0.95) {A};
  \draw[line width=1pt] (0,0.55)--(3.15,0.55)--(3.15,-0.55)--(0,-0.55);
  \node[jb, minimum height=0.80cm] at (1.3,0) {5};
  \node[jb, minimum height=0.80cm] at (2.0,0) {5};
  \node[jb, minimum height=0.65cm] at (2.7,0) {4};
  \node[server] (s1) at (3.85,0.58) {}; \node[server] (s2) at (3.85,-0.58) {};
  \node[jb, minimum height=0.35cm] at (s1) {1};
  \node[jb, minimum height=0.45cm] at (s2) {2};
\end{tikzpicture}}}
\hspace{0.1\textwidth}
\subfloat[Example under policy $B$.\label{fig:notation-ex-B}]{%
\resizebox{0.25\linewidth}{!}{\begin{tikzpicture}[server/.append style={minimum size=1.0cm}]
  \node[font=\large] at (2.0,0.95) {B};
  \draw[line width=1pt] (0,0.55)--(3.15,0.55)--(3.15,-0.55)--(0,-0.55);
  \node[jb, minimum height=0.35cm] at (2.7,0) {1};
  \node[server] (s1) at (3.85,0.58) {}; \node[server] (s2) at (3.85,-0.58) {};
  \node[jb, minimum height=0.35cm] at (s1) {1};
  \node[jb, minimum height=0.55cm] at (s2) {3};
\end{tikzpicture}}}
}
\caption{Example of two M/G/2 systems operating under scheduling policies $A$ and $B$. Each gray box represents a job, and the number in each box represents the corresponding remaining size.}
\label{fig:notation-example}
\end{figure}

\begin{table}
  {\begin{tabular}{|c|c|c|c|c|c|}
    \hline
    Policy $\pi$ & $\rb^\pi(t)$ & $N^\pi(t)$ & $d(t)$ & $\bb^\pi(t)$ & $W^\pi(t)$ \\
    \hline
    $A$ & $\rb^A(t)=(1,2,4,5,5)$ & $N^A(t)=5$ & $d(t)=2$ & $\bb^A(t)=(1,2,4,5,5)$ & $W^A(t)=17$\\
    $B$ & $\rb^B(t)=(1,1,3)$ & $N^B(t)=3$ & $d(t)=2$ & $\bb^B(t)=(0,0,1,1,3)$ & $W^B(t)=5$ \\
    \hline
  \end{tabular}}
  \caption{Example of notation for systems in \cref{fig:notation-example}.}
  \label{tab:notation-example}
\end{table}

An important concept in our proofs is a \emph{dominance} relationship between two policies:
\begin{definition}
    \label{def:dominate}
    In a coupled system of M/G/$k$ queues experiencing the same arrival process, we define policy $A$ to \emph{dominate} policy $B$ \new{at time $t$} whenever $\forall i,\; b_i^A\new{(t)} \le b_i^B\new{(t)}.$
\end{definition}

Note that we define scheduling policies which have access to the full state of the coupled system. Because all information necessary to simulate the coupled system is known, the state of the coupled system can be taken as the auxiliary Markovian state described above.

For ease of notation, when the time $t$ is clear from the context, we may omit the explicit reference $(t)$ and use $\bbb$ instead of $\bbb(t)$, for example. In the rest of this paper, we use the notation $(a)^+:=\max(a,0)$ for the positive part of a quantity $a$.

\subsection{Defining SEK-SMOD}
\label{sec:def-sek-smod}

We now introduce the SRPT-Except-$k+1$ \& Modified SRPT (SEK-SMOD) scheduling policy. This policy is built from two simpler policies, SEK and SMOD, switching between the two based on the counterfactual state of an SRPT-$k$ system, with coupled arrivals but different scheduling.

We start by defining the SEK and SMOD policies individually, and only then define SEK-SMOD. The SEK policy is a standalone policy, defined without reference to any auxiliary state.

\begin{definition}
    \label{def:sek}
    The \textit{SRPT-Except-$k+1$} (SEK) policy 
    is parameterized by constants $0 \le \epsilon' < \epsilon \le x < y$, as follows.

    SEK serves the $k$ jobs of least remaining size (matching SRPT-$k$), except when the following scenario arises:
    \begin{enumerate}
        \item There are exactly $k+1$ jobs present,
        \item $k$ of those jobs have remaining sizes between $\epsilon'$ and $\epsilon$, and
        \item the largest job has remaining size between $x$ and $y$.
    \end{enumerate}
    In this case, SEK serves the $k-1$ jobs of least remaining size, and the final job of largest remaining size until an arrival or completion occurs.
\end{definition}

We now define the SMOD policy, which requires reference to \new{an auxiliary} SRPT-$k$ system.

\begin{definition}
    \label{def:smod}
    The Modified SRPT (SMOD) scheduling policy is defined \new{in relation to auxiliary state information tracking the behavior of a hypothetical SRPT-$k$ system on the same arrival process}, that is, based on the combined state $\bbb=(\bb^{SMOD},\bb^{SRPT\hy k})$. 
    \begin{itemize}
        \item If $N^{SMOD}\leq N^{SRPT \hy k}$, \new{SMOD serves the $k$ jobs of shortest remaining size, that is, the same jobs that SRPT-$k$ would serve.}

        \item If $N^{SMOD}> N^{SRPT \hy k}$ and $N^{SMOD}>k$, \new{SMOD considers the $k+d$ jobs with smallest remaining size (recall $d := N^{SMOD} - N^{SRPT\hy k}$). We call these $k+d$ jobs \textit{eligible for service}. }
        SMOD gives \textit{high priority} to eligible jobs such that $b_i^{SMOD} - b_i^{SRPT\hy k}\geq 0$, that is, if the SMOD job's remaining size is larger than or equal to that of the SRPT-$k$ job. We call these indices \textit{zero-or-positive difference} indices.
    \begin{itemize}
        \item If there are at least $k$ zero-or-positive-difference eligible indices, SMOD serves the $k$ jobs with least remaining size among them.
        \item Otherwise, SMOD serves all the zero-or-positive-difference eligible indices, and fills the remaining servers with the least-remaining-size jobs among the negative-difference eligible indices.
    \end{itemize}
    \end{itemize}
    \new{The hypothetical SRPT-$k$ auxiliary state is defined by tracking the behavior of \new{an} SRPT-$k$ system on the same arrival process. We frequently couple the SMOD and SRPT-$k$ systems together, and do not distinguish between the coupled system and the hypothetical SRPT-$k$ auxiliary state.}
 \end{definition}

Note that all of the padding indices $i \le d$ are zero-or-positive-difference indices, and some of the remaining $k$ eligible indices, $d < i \le k+d$, may be zero-or-positive-difference as well.

\new{Intuitively, SMOD is trying to serve the same jobs that SRPT-$k$ would, as fast as possible. To do so, whenever SMOD has more jobs to process than SRPT-$k$, SMOD prioritizes eligible jobs with zero-or-positive-difference, that is, jobs with a larger remaining size than their SRPT-$k$ ``twins''.

Note that the SMOD scheduling policy only differs from SRPT-$k$ when it has more jobs to process. As soon as SMOD  reaches an identical state as its SRPT-$k$ ``twin,'' the policies match for all future states, that is, if for some $t$, $\bb^{SMOD}(t) = \bb^{SRPT\hy k}(t)$, then the policies are exactly the same for all $t'>t$. When SMOD and its SRPT-$k$ ``twin'' have identical states, we say that the systems are in ``merged states'', or that the two systems are merged.}

As shown in \cref{sec:empirical}, the Practical SEK algorithm (\cref{def:practical-sek}) exhibits excellent performance in simulations. The Practical SEK policy depends on a reduced number of parameters, simplifying the definitions above. However, it is difficult to prove that it always improves on SRPT-$k$,
as there are unlikely scenarios where not following SRPT-$k$ may incur a high cost.
To control for this, we define the SEK-SMOD policy, which ``protects'' the system's performance by only using SEK when it is provably beneficial. 

SEK-SMOD is defined in terms of \textit{divergence points}, which we define to be each moment in time that begins an interval during which the SEK-SMOD policy differs from SRPT-$k$. Further, we refer to any state in which a divergence point can occur as a \textit{divergence starting state}, that is,
a state where SEK service is in use, and where the system state satisfies the SEK divergence conditions given in \cref{def:sek}.

We are now ready to define the combined SEK-SMOD scheduling policy.

\begin{definition}
    \label{def:sek-smod}
    The SRPT-Except-$k+1$ \& Modified SRPT (SEK-SMOD) scheduling policy \new{maintains an auxiliary SRPT-$k$ system which evolves according to} the same arrival process,
    and \new{is} parameterized by constants $0 \le \epsilon' < \epsilon \le x < y$, as follows.

    \begin{itemize}
        \item If the SRPT-$k$ and SEK-SMOD systems have identical states, the SEK-SMOD policy serves jobs according to the SEK policy given in \cref{def:sek}.
        \new{In this subset of states, we call the SRPT-$k$ and SEK-SMOD systems ``merged''.}

        \item If the systems do not have identical states, and no jobs have arrived since the most recent divergence point, then SEK-SMOD continues according to the SEK policy.

        \item If the systems do not have identical states, and there has been an arrival since the most recent divergence point, then SEK-SMOD uses the SMOD scheduling policy from \cref{def:smod},
        \new{initializing the auxiliary state of the SMOD system to match SEK-SMOD's auxiliary SRPT-$k$ system}.
    \end{itemize}
\end{definition}

We focus on the following parameterization of SEK-SMOD to prove our main result, \cref{thm:main}. Under this parameterization, we show that SEK-SMOD always improves upon SRPT-$k$.
\begin{definition}
    \label{def:parameter}
    For a given $\lambda, \new{F_S},$ and $k \ge 2$ setting, we define the following SEK-SMOD parameters. We start by defining $x$ to be an arbitrary size such that $\P{S \in [x, 2x]} > 0$, and $y = 2x$.
    
    To define $\epsilon$,
    we first define the constants $c_1, c_2, c_3, c_4$, which are used in stating \cref{lem:any-diff,lem:bad,lem:good}:
    \begin{align*}
        c_1 &= 2k\lambda,
        c_2 = k \left(\lambda \frac{(k+1)y + kx/6}{1-\rho_{\le y}} + k+2\right), 
        c_3 = \frac{(\lambda kx/3)^k e^{-\lambda k(y+\new{2x})}}{k!} \P{S \in [x, 2x]}^k,
        c_4 = k/2,
    \end{align*}
    where $\rho_{\leq y}:=\lambda \E{S\ind{S\leq y}}$ is the load of jobs with size at most $y$. 
    With these constants, we define $\epsilon = \min(\frac{x}{6}, \frac{c_3 c_4}{2 c_1 c_2})$ and $\epsilon' = \frac{\epsilon}{2}$, completing the parameterization.
\end{definition}

Note that for simplicity, we specify an exact parameterization above.
However, our proof has more flexibility: $\epsilon$ may be any sufficiently small value above 0, and $y$ may be any sufficiently large value above $x$. These parameters are sufficient for our proof, and \new{we} do not optimize them in this paper.

\new{Empirically, as we discuss in \cref{sec:empirical} with the Practical SEK policy, the least mean response time is achieved when the parameters are chosen to allow as many divergences as possible: $\epsilon' = 0$, $x = \epsilon$, $y = \infty$. The best empirical choice of $\epsilon$ depends on the load $\rho$ and the job size variability. Note, however, that our proof does not apply to this choice of parameters.}

\subsection{Defining performance measures}
\label{sec:performance-measures}

A job's response time is the duration from its arrival to its completion.
We define $T^\pi$ to be a random variable denoting a generic job's response time in stationarity under the scheduling policy $\pi$.
This paper focuses on analyzing the mean response time $\E{T^\pi}$ for scheduling policies $\pi$.

To analyze the mean response time $\E{T^\pi}$, we make use of Little's law, $\lambda \E{T^\pi} = \E{N^\pi}$. Specifically, we show that $\E{T^{SEK \hy SMOD}} < \E{T^{SRPT\hy k}}$, 
or equivalently $\E{N^{SRPT\hy k} - N^{SEK \hy SMOD}} > 0$. 
To prove this, we analyze the integrated difference in number of jobs between the two systems, starting at an SEK divergence point, and ending when the SEK-SMOD and SRPT-$k$ systems are next identical.

\begin{definition}
    \label{def:ind}
    The \textit{integrated number difference} (IND) starting at a given SEK divergence point, which occurs at time $t$, is defined as follows.

    Let $\tau$ be the next time after $t$ when the SRPT-$k$ and SEK-SMOD systems have identical states, and define $\Delta_t$ as the difference in number of jobs between the SRPT-$k$ and SEK-SMOD systems, integrated over the interval $[t, \tau]$. Explicitly, we define $\Delta_t$:
    \begin{align*}
        \Delta_t := \int_{u=t}^\tau N^{SRPT\hy k}(u) - N^{SEK \hy SMOD}(u)\; du.
    \end{align*}
\end{definition}

\section{Main result}

We show that the SRPT-Except-$k+1$ \& Modified SRPT (SEK-SMOD) policy, defined in \cref{def:sek-smod}, achieves lower mean response time in the M/G/$k$ system than the previously state-of-the-art SRPT-$k$ policy, for all job size \new{distributions $F_S$}, all arrival rates $\lambda$, and all numbers of servers $k \ge 2$.
Intuitively, the SEK-SMOD policy identifies a family of situations where SRPT-$k$ performs suboptimally, and deviates away from SRPT-$k$ only in those situations.

These deviations occur when there is a tradeoff between immediately serving jobs with low remaining size, and saving \new{some of these} jobs to increase utilization in the future. SRPT-$k$ always prioritizes the smallest jobs, while SEK-SMOD selectively prioritizes future utilization. We demonstrate that the balance of probabilities favors prioritizing future utilization in the situations we identify, allowing SEK-SMOD to outperform SRPT-$k$.

Our main result is as follows:

\begin{theorem}
    \label{thm:main}
    For all arrival rates $\lambda$, job size \new{distributions $F_S$}, and numbers of servers $k\geq 2$,
    with the SEK-SMOD parameterization given in \cref{def:parameter},
    the mean response time of SEK-SMOD is smaller than that of SRPT-$k$:
    \begin{align*}
        \E{T^{SEK \hy SMOD}} < \E{T^{SRPT \hy k}}. 
    \end{align*}
\end{theorem}
We defer the proof  to \cref{sec:main-proof}.

Before providing the detailed proof, we provide an outline in \cref{sec:proof-outline}. The proof is broadly split into stochastic and worst-case \new{(i.e. sample-path)} arguments. The stochastic arguments focus on the SEK portion of the policy, and the corresponding lemmas are proved in \cref{sec:stochastic}. The worst-case proofs focus on the SMOD portion of the policy, and are presented in \cref{sec:worst-case}.

\new{We empirically evaluate the magnitude of improvement of SEK over SRPT-$k$ in \cref{sec:empirical},
and discuss practical considerations for the implementation and use of SEK in \cref{sec:discuss}.}

\subsection{Proof outline}
\label{sec:proof-outline}

In this section, we discuss the key ideas in our proof, which are also summarized in \cref{diagram:thm-proof}. To prove that the SEK-SMOD policy always improves upon the SRPT-$k$ policy, we take a divergence-by-divergence approach. \new{To do so, we show in \cref{sec:main-proof} that} the M/G/$k$/SEK-SMOD system is positive recurrent, and that the empty system defines a renewal process. 

\begin{figure}
    \mbox{}
    \hspace{-10pt}
    {\begin{tikzpicture}[
        node distance=0.3cm and 0.3cm,
        box/.style={rectangle, draw, align=center, text width=4cm},
        bullet-box/.style={rectangle, draw, align=center, text width=4cm},
        worst-box/.style={rectangle, fill=gray!30, draw=none, text width=6cm, align=center},
        stoch-box/.style={rectangle, fill=gray!30, draw=none, text width=5cm, align=center},
        edge/.style={draw, -{Latex}},
    ]
    \begin{scope}[font=\footnotesize]
    \node[box] (main) {\textbf{T\ref{thm:main}:} SEK-SMOD outperforms SRPT-$k$};
    \node[box, below=of main] (any-diff) {\textbf{L\ref{lem:any-diff}:} Positive expected improvement per divergence.};
    \node[box, right=of any-diff] (sek-bridge) {\textbf{L\ref{lem:sek-bridge}} implies empty SEK-SMOD system is a renewal process.};
    \node[bullet-box, below=of any-diff] (neutral) {\textbf{L\ref{lem:neutral}: Neutral scenario} \\ Expected improvement is nonnegative}; 
    \node[bullet-box, left=of neutral] (bad) {\textbf{L\ref{lem:bad}: Bad scenario}\\[4pt]
        \begin{minipage}{\textwidth}
        $\bullet$ Expected improvement: $-O(\epsilon)$\\
        $\bullet$ Probability: $O(\epsilon)$
        \end{minipage}}; 
    \node[bullet-box, right=of neutral] (good) {\textbf{L\ref{lem:good}: Good scenario} \\[4pt]
        \begin{minipage}{\textwidth}
        $\bullet$ Expected improvement: $\Omega(\epsilon)$\\
        $\bullet$ Probability: $\Omega(1)$
        \end{minipage}};
    \node[stoch-box, below=of good] (stoch) {Stochastic arguments (see \Cref{sec:stochastic})};  
    \node[worst-box, left=of stoch] (worst) {Worst-case arguments (see \Cref{sec:worst-case})};
    \path[edge] (sek-bridge) -- (main);
    \path[edge] (any-diff) -- (main);
    \path[edge] (neutral) -- (any-diff);
    \path[edge] (bad) -- (any-diff);
    \path[edge] (good) -- (any-diff);
    \path[edge] (worst) -- (bad);
    \path[edge] (worst) -- (neutral);
    \path[edge] (stoch) -- (good);
    \end{scope}
    \end{tikzpicture}}
\caption{Diagram of proof of Theorem \ref{thm:main}. Flowchart of the proof: Worst-case arguments in \cref{sec:worst-case} underpin results for the bad and neutral scenarios, while stochastic arguments underpin results for the good scenario. These results are combined to prove our main result.}
\label{diagram:thm-proof}
\end{figure}

Note that a positive expected IND is equivalent to mean response time improvement, as we show in \cref{sec:main-proof}. We therefore focus on characterizing the effects of each SEK divergence on the expected IND. Specifically, we show that the expected IND is positive on each SEK divergence from SRPT-$k$ (\cref{lem:any-diff}).

Once SEK diverges from SRPT-$k$, an arrival happening too soon may result in negative IND. However, this happens with a small probability and is compensated by the gains when there are no arrivals for a relatively short period of time. \cref{lem:any-diff} formalizes this intuition by distinguishing three scenarios. For each of them, we bound the expected IND and frequency to obtain our main result. We provide an intuition to these scenarios below, and a formal definition in \cref{sec:stochastic}:
\begin{itemize}
    \item A \textit{good scenario} is when no jobs arrive for a substantial period of time after a divergence point. Then, the schedule selected by SEK-SMOD benefits the IND for sufficiently long, and provides an improvement of $\Omega(\epsilon)$. Because arrivals occur according to a Poisson process, this event happens with a positive probability, lower bounded by a constant. See \cref{lem:good}.
    \item A \textit{bad scenario}, on the contrary, happens when a job arrives immediately after a divergence point. This situation may yield a decrease in IND by $O(\epsilon)$, but it is rare and, indeed, happens with probability $O(\epsilon)$. Hence, a bad scenario contributes with a negative effect of at most $O(\epsilon^2)$. See \cref{lem:bad} for a formalization of this result and the corresponding proof.
    \item All other possibilities are called a \textit{neutral scenario}, and we show that they have a nonnegative effect on the expected IND. See \cref{lem:neutral}.
\end{itemize}

We specifically study these scenarios as $\epsilon \to 0$, limiting the size of the $k$ smallest jobs at the divergence point to smaller and smaller sizes.
We demonstrate that the advantage of SEK-SMOD policy over the SRPT-$k$ policy is clearest as this $\epsilon$ gets smaller, and specifically, that there exists an explicit $\epsilon$ threshold small enough such that the expected IND is always positive, given in \cref{def:parameter}.

The stochastic results in the three scenarios are underpinned by worst-case results for the SMOD policy, especially in the bad scenario.
We provide a discussion and a formal proof of these worst-case results in \cref{sec:worst-case}.

Our worst-case results allow us to shift our attention from the duration until the SEK-SMOD system is \emph{identical} to the SRPT-$k$ system, to the duration until the SEK-SMOD system \emph{dominates} SRPT-$k$ system (see \cref{def:dominate}). Intuitively, when SEK-SMOD dominates SRPT-$k$, the SEK-SMOD system state is strictly preferable to that of the SRPT-$k$ system. We prove in \cref{cor:dominance} that this property is maintained until the next divergence point. We then bound the time until dominance in \cref{lem:max-diff}, and in \cref{lem:diff-per-job} we bound the IND in terms of the number jobs that arrive before dominance.
These combine to give the IND bounds that we need.

These worst-case bounds are incorporated into our stochastic results for the good, bad, and neutral scenarios, demonstrating the desired behavior in the $\epsilon \to 0$ limit. This behavior implies positive expected IND for any $\epsilon$ smaller than an explicit threshold, which in turn implies our main result, \cref{thm:main}: SEK-SMOD with the parameterization given in \cref{def:parameter} always improves on SRPT-$k$.

\section{Stochastic proofs}
\label{sec:stochastic}

In this section, we prove our main result, \cref{thm:main}.
To do so, we prove that each SEK-SMOD divergence has a positive expected IND. The timing of the next arrival after a divergence point and the corresponding job size have a significant influence on the number of jobs present in each system at each time interval, determining the value of the IND.  
\new{As mentioned in \cref{sec:proof-outline}, \cref{lem:any-diff} identifies three possible scenarios and studies the corresponding IND separately. We start by formally defining them:}
\begin{definition}\label{def:good-bad-neutral}
    Define three events in the time immediately following a divergence point:
    \begin{description}
        \item[Bad scenario:] A job arrives during the next $2k\epsilon$ time after the divergence point, before all small jobs are guaranteed to finish under SEK-SMOD.
        \item[Good scenario:] No job arrives during the next $kx/3 \ge 2k\epsilon$ time after the divergence point,
        and then $k$ big enough jobs arrive. These jobs are large enough to wait in line until the largest job in service completes.
        Specifically, letting the size of the largest job at time of divergence be $\new{\ell}$, we require that exactly $k$ jobs with sizes in $[x, 2x]$ arrive during the window $[k(\new{\ell}-2x/3), k(\new{\ell}-x/3)]$, and no other jobs arrive before all $k$ of those jobs are done and the originally-largest job is done, at time $k(\new{\ell}+2x)$.
        \item[Neutral scenario:] Everything else.
    \end{description}
    We define the events $\bad_t, \good_t,$ and $\neutral_t$ to occur when the corresponding scenario occurs due to a divergence at time $t$.
\end{definition}

In \cref{lem:bad,lem:good,lem:neutral}, respectively,
we lower bound the IND in each scenario, based on the stochastic properties of the SEK policy.
These stochastic results are underpinned by worst-case results on the SMOD policy,
which we prove in \cref{sec:worst-case}.
Finally, in \cref{sec:main-proof}, we bring these results together to prove our main result, \cref{thm:main}.

\begin{lemma}
    \label{lem:any-diff}
    For any $\lambda, \new{F_S}, k$, under the SEK-SMOD parameterization in \cref{def:parameter},
    and for any divergence starting state $\bb\bb$, 
    the integrated number difference is positive. Specifically,
    \begin{align*}
        \E{\Delta_t \mid \bbb(t) = \bbb} \geq \epsilon c_3c_4 - \epsilon^2 c_1c_2,
    \end{align*}
    which is positive whenever the parameter $\epsilon<\frac{c_3c_4}{c_1c_2}$.
\end{lemma}

Note that \cref{lem:any-diff} provides a lower bound on the performance improvement of SEK-SMOD compared to SRPT-$k$. In \cref{sec:empirical} we show that the practical improvement may be considerably higher, even for a simplified version of the SEK-SMOD policy.

\begin{proof}{Proof of \cref{lem:any-diff}}
    To compute a lower bound on $\E{\Delta_t \mid \bbb(t) = \bbb}$, we condition on the possible scenarios. We show the following lower bounds for any divergence starting state $\bbb$, where $c_1, c_2, c_3, c_4$ are defined in \cref{def:parameter}. 
    We focus on proving lower bounds that are tight in the $\epsilon\to 0$ limit,  for any divergence starting state $\bbb$. Specifically, we show:
    \begin{description}
        \item[Bad:] (\cref{lem:bad}) $P(\bad_t \mid \bbb(t) = \bbb) \le c_1 \epsilon$, and $\E{\Delta_t \mid \bbb(t) = \bbb \,\&\, \bad_t} \ge - c_2 \epsilon$.
        \item [Good:] (\cref{lem:good}) $P(\good_t \mid \bbb(t) = \bbb) \ge c_3$, and $\E{\Delta_t \mid \bbb(t) = \bbb \,\&\, \good_t} \ge c_4 \epsilon$.
        \item [Neutral:] (\cref{lem:neutral}) $\E{\Delta_t \mid \bbb(t) = \bbb \,\&\, \neutral_t} \ge 0$.
    \end{description}
    Recall from \cref{def:parameter} that the constants $c_1, c_2, c_3, c_4$ depend on $\lambda, \new{F_S}, k$, but not on $\epsilon$.

    Then, by the law of total probability and using the bounds listed above, we see that
    \begin{align*}
        & \E{\Delta_t \mid \bbb(t) = \bbb} \geq -\epsilon^2 c_1c_2 + \epsilon c_3 c_4 + 0.
    \end{align*}
    
    With these lower bounds, we see that for any $\epsilon < \frac{c_3 c_4}{c_1 c_2}$, as in the parameterization in \cref{def:parameter}, we have $\E{\Delta_t \mid \bbb(t)=\bbb} > 0$ for any divergence starting state $\bbb$, as desired. 
\end{proof}

We now formalize the results to lower bound the IND for the bad, good, and neutral scenarios:

\begin{lemma}[Bad scenario]
    \label{lem:bad}
    \new{There exist} constants $c_1, c_2 > 0$ (defined in \cref{def:parameter}) that depend on $\lambda, \new{F_S}, k$ and not on $\epsilon$, \new{such that} for any SEK-SMOD parameterization as in \cref{lem:any-diff} and any divergence starting state $\bbb$, \new{the following holds}:
    \begin{align}
    \label{eq:bad-prob}
        P(\bad_t \mid \bbb(t) = \bbb) &\le c_1 \epsilon,    \\
    \label{eq:bad-ind}
        \E{\Delta_t \mid\bbb(t) = \bbb \,\&\, \bad_t} &\ge - c_2 \epsilon.
    \end{align}
\end{lemma}
\begin{proof}{Proof of \cref{lem:bad}}
    \new{We} start with \eqref{eq:bad-prob}. 
    Note that for the bad event to happen,
    a job must arrive in the next $2k\epsilon$ time.
    This occurs with probability $1 - e^{-{2k\lambda \epsilon }}$, as the arrival process is a Poisson process with rate $\lambda$. Note that $1 - e^{-{2k\lambda\epsilon}} \le 2k\lambda \epsilon$. 
    Thus, because $c_1 = 2k\lambda$, we confirm \eqref{eq:bad-prob}.

    Next, \new{we} turn to \eqref{eq:bad-ind}.
    To prove this, we invoke the worst-case properties of SEK-SMOD proven in \cref{sec:worst-case}.
    \new{We show in \cref{lem:diff-per-job,lem:sek-bridge} a bound on $\E{t_{dom} \mid\bbb(t) = \bbb \,\&\, \bad_t}$ that depends on the completion time of the largest job present at the time of divergence.
    This completion time can in turn be upper bounded using multiserver tagged-job analysis of \cite{grosof-etal-2018-srpt-k}, which completes the proof.} We provide the details in \cref{ap:proof-bad}. \end{proof}

\begin{lemma}[Good scenario]
    \label{lem:good}
    \new{There exist} constants $c_3, c_4 > 0$ (defined in \cref{def:parameter}) that depend on $\lambda, \new{F_S}, k$ and not on $\epsilon$, \new{such that for any SEK-SMOD parameterization as in \cref{lem:any-diff} }and any divergence starting state $\bbb$, \new{the following holds:}
    \begin{align*}
        & P(\good_t \mid \bbb(t) = \bbb) \ge c_3, \quad
        \E{\Delta_t \mid \bbb(t) = \bbb \,\&\, \good_t} \ge c_4 \epsilon.
    \end{align*}
\end{lemma}
\begin{proof}{Proof of \cref{lem:good}}
    First, to characterize $P(\good_t \mid \bbb(t) = \bbb)$,
    note that the good scenario described in \cref{def:good-bad-neutral}
    does not depend on $\epsilon$ at all.
    We compute its probability and verify it \new{is} at least $c_3$ in \cref{ap:good-probabilities}.
    
    We now bound $\E{\Delta_t \mid \bbb(t) = \bbb \,\&\, \good_t}$.
    Let us compare the behavior under SRPT-$k$ with that under SEK-SMOD. 
    The rest of the proof can be summarized in the following situations, depending on which jobs are being counted in the IND. Each of them has an effect $\Delta_t$, characterized below.
    \begin{itemize}
        \item When counting the $k-1$ jobs with least remaining size at time of divergence, both systems have the same completion time.
        \item When counting the second-largest job at time of divergence, SRPT-$k$ completes earlier by $k b_1$, where $b_1$ is the remaining size of the smallest job at time of divergence, under $\bbb$.
        \item When counting the largest job at the time of divergence, SEK-SMOD completes earlier by $k b_1$.
        \item If there are later arriving jobs before merging, then $k-1$ of the jobs have the same completion times,
        and one has a completion time $k b_1$ earlier in the SEK-SMOD system.
    \end{itemize}
    Note that $\Delta_t$ is equivalent to the difference in total response time over all jobs that arrive after the divergence at time $t$ and prior to merging. Consequently, we show that $c_4 = k/2$ and 
    \begin{align*}
        & \E{\Delta_t \mid \bbb(t) = \bbb \,\&\, \good_t} = k b_1 \ge k \epsilon/2.
    \end{align*}

    The proof of each bullet point is simple, as it only requires careful analysis of the remaining job sizes. We provide the details in \cref{ap:good-expectation}. 
    \end{proof}

\begin{lemma}[Neutral scenario]
    \label{lem:neutral}
    For the SEK-SMOD parameterization given in \cref{def:parameter}
    and any \new{divergence} starting state $\bbb$,
    all neutral and good scenarios achieve nonnegative $\Delta_t$. \new{In other words,} we always have
    \begin{align*}
        [\Delta_t \mid \bbb(t) = \bbb \,\&\, (\good_t \vee \neutral_t)] \ge 0.
    \end{align*}
\end{lemma}
\begin{proof}{Proof of \cref{lem:neutral}}
    In the good and neutral scenarios, no job arrives in the next $2k\epsilon$ time after the divergence point, which occurs at time $t$.
    The integrated difference from time $t$ to $t+2k\epsilon$ is $-kb_1$, where $b_1$ is
    the remaining size of the smallest job in $\bbb$ at time of divergence:
    \begin{align}
        \label{eq:neutral-first-interval}
        \int_{u=t}^{t+2k\epsilon} \left( N^{SRPT\hy k}(u) - N^{SEK \hy SMOD}(u)\right) \, du= -kb_1.
    \end{align}
    To show this, we can examine the completion times or lack thereof of each of the $k+1$
    jobs which are present at divergence under $\bbb$,
    up to time $t+2k\epsilon$. 

    To complete the proof, we analyze the IND under three possible situations:

\begin{itemize}
        \item The $k-1$ jobs of smallest remaining size all are in service at the divergence time under both policies, and finish at the same time under both policies.
        This has no effect on the integrated difference.
        \item The job of second-largest remaining size is in service at the divergence time $t$ under the SRPT-$k$ policy,
        and finishes at time $t+kb_k$.
        Under the SEK-SMOD policy, the same job only enters service at time $t+k b_1$,
        when the job of smallest remaining size completes,
        and finishes at time $t+k(b_1 + b_k)$. Note that $b_1 \le b_k \le \epsilon$, so $k(b_1 + b_k) \le 2k \epsilon$, so this completion occurs in the desired time window.
        This generates the $-kb_1$ difference described above.
        \item The job of largest remaining size has not finished by time $2k\epsilon$
        under either policy, as $b_{k+1} \ge x \ge 6\epsilon$.
        This has no effect on the integrated difference.
    \end{itemize}

    At time $t+2k\epsilon$,
    the state of the SRPT-$k$ system is a single job with remaining size $b_{k+1} - (2\epsilon - b_1)$,
    as the largest job entered service at time $t+kb_1$, and has received $(2\epsilon - b_1)$ service.
    The state of the SEK-SMOD system is a single job with remaining size $b_{k+1} - 2\epsilon$,
    as the largest job entered service at time $t$ and has received $2\epsilon$ service.
    In each system, all other jobs have completed and no jobs have arrived.

    Note in particular that SEK-SMOD dominates SRPT-$k$ at this point in time,
    because the SEK-SMOD job has \new{smaller} remaining size. 
    We now invoke \cref{lem:negative-bounds-diff,lem:sek-bridge}, which together show that the IND
    over the interval from $t+2k\epsilon$ to the time of the next merge $\tau$,
    \begin{align}
        \label{eq:neutral-second-interval}
        \int_{u=t+2k\epsilon}^{\tau} \left( N^{SRPT\hy k}(u) - N^{SEK \hy SMOD}(u)\right) \, du \ge kb_1.
    \end{align}

    Combining the intervals in \eqref{eq:neutral-first-interval} and \eqref{eq:neutral-second-interval},
    we find that $\Delta_t \ge 0$, always. 
    \end{proof}

\subsection{Proof of Theorem \ref{thm:main}}
\label{sec:main-proof}

We are now ready to prove \cref{thm:main}, our main result. \new{We repeat the statement for convenience.}

\begin{reptheorem}{thm:main}
    For all arrival rates $\lambda$, \new{job size distributions $F_S$}, and numbers of servers $k\geq 2$,
    with the SEK-SMOD parameterization given in \cref{def:parameter},
    the mean response time of SEK-SMOD is smaller than that of SRPT-$k$:
    \begin{align*}
        \E{T^{SEK \hy SMOD}} < \E{T^{SRPT \hy k}}. 
    \end{align*}
\end{reptheorem}

\begin{proof}{Proof of \cref{thm:main}}
    Recall from \cref{sec:performance-measures} that by Little's Law it is equivalent to prove that $\E{N^{SRPT \hy k} - N^{SEK \hy SMOD}} > 0$, which is our focus.
    
    We analyze the behavior of the \new{coupled} (SEK-SMOD, SRPT-$k$) system as a series of cycles.
    First, SEK-SMOD diverges from SRPT-$k$. Next, the two systems eventually \textit{merge}, that is, they have the same state once more. Finally, another cycle begins with another divergence.
    Note that the difference $N^{SRPT \hy k} - N^{SEK \hy SMOD}$ can only be nonzero when the system states are different.

    To analyze the behavior of the \new{coupled} system, we use a renewal-reward argument.
    To do so, we need a renewal point that has finite expected time between renewals.
    Our renewal points are the times when both systems are empty, and a job arrives.
    \new{First, note that every work-conserving policy in the M/G/$k$ has a finite mean time between visits to the empty state, for any parameterization $\lambda, S,$ and $k$ with $\rho = \lambda \E{S} < 1$. This is
    a folk result which we prove in \cref{thm:pos-recurrence} in Appendix~\ref{app:pos-recurrence} for completeness.}
    \new{As a result,} the SRPT-$k$ system has finite mean time between visits to the empty state.
    We need to show that there is a finite mean time between visits to the jointly empty state \new{in the coupled system}.

    We prove this by reference to our worst-case lemmas.
    \new{We show in \cref{lem:sek-bridge} that} the work under SEK-SMOD never exceeds work under SRPT-$k$: $W^{SEK \hy SMOD} \le W^{SRPT\hy k}$. Therefore, an empty SRPT-$k$ system ($W^{SRPT\hy k} = 0$) implies $W^{SEK \hy SMOD} = 0$, and thus the SEK-SMOD system is empty as well. Hence, the jointly empty state is a valid renewal point with finite expected time per renewal.

    We apply the renewal reward formula, where our reward function is $N^{SRPT\hy k} - N^{SEK \hy SMOD}$:
    \begin{align*}
        &\E{N^{SRPT \hy k} - N^{SEK \hy SMOD}} 
        = \frac{\E{\int_{u \in \text{cycle}} N^{SRPT \hy k}(u) - N^{SEK \hy SMOD}(u) du}}{\E{\text{Length of cycle}}}.
    \end{align*}
    Our goal is to show that the numerator, the expected reward per both-empty cycle, is positive.
    We can quantify this total reward per both-empty cycle in terms of the expected total reward per divergence event in a given cycle, that is,
    \begin{align*}
        &\E{\int_{u \in \text{cycle}} N^{SRPT \hy k}(u) - N^{SEK \hy SMOD}(u) du} 
        = \E{\sum_{t \in \text{cycle divergence points}} \Delta_t}.
    \end{align*}
    Note that this sum is valid because $\tau$, the merging time at which $\Delta_t$ completes,
    is bounded by the time at which the jointly empty state next occurs.
    We show that this expected sum is positive.

    \cref{lem:any-diff} proves that $\E{\Delta_t \mid \bb\bb(t) = \bb\bb} > 0$,
    for any arbitrary divergence starting state $\bb\bb$. Observe that the number of divergences per cycle is irrelevant because each divergence contributes a positive expectation, and there is a positive expected number of divergences per cycle. Additionally, because we condition on an arbitrary starting state, correlations between different divergences are also irrelevant to the analysis.
    
    This result implies that the expected total reward per \new{jointly}-empty-cycle is positive, because the probability of a divergence event is positive. Therefore, $\E{T^{SEK \hy SMOD}} < \E{T^{SRPT\hy k}}$ as desired. 
    \end{proof}

\section{Worst case proofs}
\label{sec:worst-case}

In this section, we prove key worst-case properties of the SMOD policy and the SEK-SMOD policy,
which are the foundational results that underpin our stochastic results in \cref{sec:stochastic}.
\new{These properties hold on all sample paths.}
Specifically, in \cref{sec:smod-wc},
we characterize the worst-case behavior of the SMOD policy relative to an SRPT-$k$ system
experiencing the same arrival sequence.
Then, in \cref{sec:sek-smod-wc},
we specialize these results to characterize the combined SEK-SMOD policy,
again relative to the SRPT-$k$ policy.

\subsection{SMOD worst case lemmas}
\label{sec:smod-wc}
In this section, we analyze the behavior of the joint system containing an SMOD queue and an SRPT-$k$ queue
under an arbitrary stream of arrivals.
Specifically, we prove a series of monotonicity and preservation lemmas,
showing that the SMOD policy monotonically improves (or preserves) \new{certain metrics} relative to the SRPT-$k$ system, 
under arbitrary arrivals. We also define and prove properties of the IND under these worst-case arrivals.

\begin{figure}
    \mbox{}
    \hspace{-20pt}
    \begin{tikzpicture}[
        node distance=0.2cm and 0.2cm,
        box/.style={rectangle, draw, align=center, text width=4cm},
        small-box/.style={rectangle, draw, align=center, text width=3cm},
        bullet-box/.style={rectangle, draw, align=center, text width=4.5cm},
        wide-box/.style={rectangle, draw, align=center, text width=5cm},
        edge/.style={draw, -{Latex}},
    ]
    \begin{scope}[font=\footnotesize]
    \node[bullet-box] (bad) {\textbf{L\ref{lem:bad}: Bad scenario.}\\[4pt]
        \begin{minipage}{\textwidth}
        $\bullet$ Expected improvement: $-O(\epsilon)$\\
        $\bullet$ Probability: $O(\epsilon)$
        \end{minipage}}; 
    \node[box, below=of bad, yshift=-.2cm] (max-diff) {\textbf{L\ref{lem:max-diff}}: Dominance by time when largest job at divergence completes.}; 
    \node[box, below=of max-diff, yshift=-.2cm] (pln) {\textbf{L\ref{lem:pln}:} PLN, work inequality and zig-zag preserved.} ;
    \node[box, anchor=north, below=of pln, yshift=-.2cm] (pos-part) {\textbf{L\ref{lem:pos-part}:} Positive part~$r^+$ weakly decreasing. \\
    \textbf{C\ref{cor:dominance}:} Dominance preserved.};
    \node[box, anchor=north, left=of pos-part] (zig-zag) {\textbf{L\ref{lem:zig-zag}: Zig-zag matching.} \\
    SMOD has at most one extra job.};
    \node[box, left=of pln] (diff-per-job) {\textbf{L\ref{lem:diff-per-job}}: IND bounded by $k\epsilon$ per job until dominance.}; 
    \node[wide-box, right=of pln] (sek-bridge) 
    {\textbf{L\ref{lem:sek-bridge}: SEK-SMOD bridge.} \\
    SEK-SMOD maintains work inequality, PLN, zig-zag and bounded positive part.} ;
    \path[edge] (diff-per-job) |- (bad);
    \path[edge] (max-diff) -- (bad);
    \path[edge] (pln) -- (max-diff);
    \path[edge] (pln) -- (sek-bridge);
    \path[edge] (zig-zag) -- (diff-per-job);
    \path[edge] (zig-zag) -- (pln);
    \path[edge] (pos-part) -- (diff-per-job);
    \path[edge] (pos-part) -- (pln);
    \path[edge] (pos-part) -- (sek-bridge);
    \path[edge] (sek-bridge) |- (bad);
    \end{scope}
    \end{tikzpicture}
    \hspace{-20pt}
    \mbox{}
\caption{Diagram of worst-case scenario arguments.
The positive part (\cref{def:positive-part}), zig-zag matching (\cref{def:zig-zag}), and PLN (\cref{def:pln}) are defined later in this section; and ``work inequality'' means that $W^A(t)\leq W^{SRPT\hy k}(t)$, where policy $A$ can be SMOD or SEK-SMOD depending on the context.}
\label{diagram:worst-case}
\end{figure}

We provide a diagram representing how these lemmas fit in the proof of \cref{thm:main} in \cref{diagram:worst-case}. The diagram refers to new concepts that are key to our proof and are defined later in this section. Specifically, we refer to the positive part (\cref{def:positive-part}), zig-zag matching (\cref{def:zig-zag}), and PLN (\cref{def:pln}). When we mention the ``work inequality'' property, we mean that $W^A(t)\leq W^{SRPT\hy k}(t)$, where policy $A$ can be SMOD or SEK-SMOD depending on the context.

\subsubsection{Positive part of work difference}
\label{sec:positive-part}

Intuitively, we \new{prove that, under an appropriately chosen metric, the amount by which} SMOD is in a worse state than SRPT-$k$ never increases. As a corollary, \new{we show that} dominance of SMOD over SRPT-$k$ is preserved into the future.

\begin{definition}\label{def:positive-part}
    Define $r^+(t)$, the \textit{positive part} of the difference in work between the SMOD and SRPT-$k$ systems,
    to be
    \begin{align*}
        r^+(t) := \sum_{i} \big(b_i^{SMOD}(t) - b_i^{SRPT\hy k}(t)\big)^+.
    \end{align*}
\end{definition}

Now, we can prove our first monotonicity result.
\begin{lemma}[Positive part weakly decreasing]
\label{lem:pos-part}
    In the SMOD-SRPT joint system,
    under an arbitrary sequence of arrivals,
    for any pair of times $t_1 \le t_2$,
    the positive part at time $t_1$ upper bounds the positive part at time $t_2$:
    \begin{align*}
        r^+(t_1) \ge r^+(t_2).
    \end{align*}
\end{lemma}
\begin{proof}{Proof of \cref{lem:pos-part}}
    There are three ways the SMOD-SRPT joint system can change:
    Jump change on arrival, jump change on completion, and continuous change due to work being completed.
    We show that each of these changes maintains the monotonicity property by carefully analyzing their effect on the remaining size and index of each of the jobs present right before and right after the event happens. The proof is \new{straightforward} once we split into the three possible events described above, so we provide the details in \cref{ap:pos-part}. 
    \end{proof}

As a corollary, SMOD preserves the dominance relationship defined in \cref{def:dominate}:
\begin{corollary}[Dominance preserved]
    \label{cor:dominance}
    If at time $t_1$, SMOD dominates SRPT-$k$, then for all $t_2 \ge t_1$,
    SMOD dominates SRPT-$k$.
\end{corollary}
\begin{proof}{Proof of \cref{cor:dominance}}
    The proof immediately follows from \cref{lem:pos-part}.
    Note that SMOD dominating SRPT-$k$ is equivalent to $r^+ = 0$.
    Therefore at time $t_1$, $r^+(t_1) = 0$ and by \cref{lem:pos-part}, $r^+(t_2) = 0$. 
    \end{proof}

\subsubsection{Zig-zag matching}
\label{sec:zig-zag}

Intuitively, if SMOD is only one job worse than SRPT-$k$, it will always remain at most one job worse. \new{We formalize this intuition with the following definition and lemma.}

\begin{definition} 
    \label{def:zig-zag}
Define the SMOD and SRPT-$k$ systems to be \emph{zig-zag matched} if:
\begin{align*}
    \forall i, b_i^{SMOD} \le b_{i+1}^{SRPT\hy k}.
\end{align*}
\end{definition}

In other words, the SMOD and SRPT-$k$ systems are zig-zag matched if, after removing the largest job in the SMOD system, it dominates the SRPT-$k$ system. 
We now show that zig-zag matching is always preserved.

\begin{lemma}[Zig-zag]
    \label{lem:zig-zag}
    If at time $t_1$, the systems are zig-zag matched,
    then at all times $t_2 \ge t_1$, the systems will remain zig-zag matched.

    As a result, the SMOD system has at most one additional job compared to the SRPT-$k$ system: $N^{SMOD}(t_2) - N^{SRPT\hy k}(t_2) \le 1$.
\end{lemma}

\begin{proof}{Proof of \cref{lem:zig-zag}}
    First, we show that zig-zag matching implies that there is at most one more job in the SMOD system than in the SRPT-$k$ system.
    For contradiction, suppose that the SMOD system has at least two more jobs than the SRPT-$k$ system.
    In this case, indices 1 and 2 in the SRPT-$k$ vector are padding indices, with $b^{SRPT\hy k}_2 = 0$.
    On the other hand, there are no padding indices in the SMOD vector,
    so $b^{SMOD}_1 > 0$. This violates zig-zag matching.
    Thus, zig-zag matching implies that there is at most one more job in the SMOD system than in the SRPT-$k$ system.

    Now, we turn to proving that zig-zag matching is preserved.
    As in the proof of \cref{lem:pos-part},
    there are three different ways that the system can change: jump change on arrivals, jump change in service, and continuous change in service. The rest of the proof carefully analyzes the job indices affected by each event, similarly to \cref{lem:pos-part}. We provide the details in \cref{ap:zig-zag-proof}. 
    \end{proof}

\subsubsection{Integrated difference per job}
\label{sec:diff-per-job}

We bound SMOD's IND relative to SRPT-$k$ by examining the time when SMOD achieves dominance, and the difference after dominance. \new{We start with a formal definition, and then state and prove the result.}

\begin{definition}\label{def:future-IND}
    Define the total \textit{future integrated number difference} $\Delta_{SRPT\hy k, SMOD}(t, \infty)$ between the SMOD and SRPT-$k$ system as follows:
    \begin{align*}
        & \Delta_{SRPT\hy k, SMOD}(t, \infty) 
        := \int_{u=t}^\infty N^{SRPT\hy k}(u) - N^{SMOD}(u) du.
    \end{align*}

    Additionally, define $A(t_1, t_2)$ to be the number of arrivals over the interval of time $[t_1, t_2]$.
\end{definition}

With the above worst-case results, we can bound the total future integrated number difference in terms of the number of arrivals:

\begin{lemma}
    \label{lem:diff-per-job}
    If at time $t_1$ the SMOD system has a positive part $r^+(t_1) \le \epsilon$ and zig-zag matching,
    and at some later time $t_{dom}$ the SMOD system dominates the SRPT-$k$ system,
    then the IND over the entire future after $t_1$ is lower bounded 
    by $k \epsilon$ times the number of jobs that are present at time $t_1$ or arrive in the interval $[t_1, t_{dom}]$. Formally,
    \begin{align*}
        &\Delta_{SRPT\hy k, SMOD}(t_1, \infty) 
        \ge - k \epsilon (N^{SMOD}(t_1) + A(t_1, t_{dom})).
    \end{align*}
\end{lemma}

\begin{proof}{Proof of \cref{lem:diff-per-job}}
    Note that the integrand of $\Delta_{SRPT\hy k, SMOD}(t, \infty)$ in \cref{def:future-IND}, namely $N^{SRPT\hy k}(u) - N^{SMOD}(u)$,
    simply measures the number of excess jobs in one system as compared to the other.
    As we want to lower bound $\Delta_{SRPT\hy k, SMOD}(t_1, \infty)$,
    we focus on upper bounding $N^{SMOD}(u) - N^{SRPT\hy k}(u)$.
    By \cref{lem:zig-zag}, this difference is at most 1 at all times.
    We want to bound how frequently this difference equals 1.

    Let us focus on the smallest job in the SMOD system at any time $t \ge t_1$ such that $N^{SMOD}(t) - N^{SRPT\hy k}(t) = 1$.
    At such a time, there must be a padding index at the SRPT-$k$ system, so $b^{SRPT\hy k}_1 = 0$.
    Thus, a positive difference must exist in the first index: $r^+(t) \ge b^{SMOD}_1$.
    But by \cref{lem:pos-part}, $r^+(t) \le r^+(t_1) \le \epsilon$.

    Therefore, the smallest job in the SMOD system must have a remaining size of at most $\epsilon$
    whenever $N^{SMOD}(t) - N^{SRPT\hy k}(t) = 1$.
    Note that this job also receives service under $SMOD$ at rate $1/k$.
    The first index is definitely in the eligible set,
    and a positive difference $b^{SMOD}_1 - b^{SRPT\hy k}_1$ must exist.
    Among all positive-difference indices in the eligible set,
    the first index has the least remaining size.
    Thus, by \cref{def:smod}, the first index is guaranteed to be served.

    Any given job can spend at most $k \epsilon$ time in service with remaining size at most $\epsilon$. 
    Thus, the total amount of time $t \ge t_1$ for which $N^{SMOD}(t) - N^{SRPT\hy k}(t) = 1$
    is bounded by $k \epsilon$ times the number of distinct jobs that are served while $N^{SMOD}(t) - N^{SRPT\hy k}(t) = 1$.
    Note that after time $t_{dom}$, $N^{SMOD}(t) - N^{SRPT\hy k}(t) = 1$ never again occurs, by \cref{cor:dominance}.
    Then, the only jobs that can be served while $N^{SMOD}(t) - N^{SRPT\hy k}(t) = 1$
    are jobs which are either present at time $t_1$ or arrive before $t_{dom}$. Thus $N^{SMOD}(t_1) + A(t_1, t_{dom})$ jobs are eligible to be served while $N^{SMOD}(t) - N^{SRPT\hy k}(t) = 1$,
    each for at most $k \epsilon$ time. 
    \end{proof}

We have another minor lemma,
based on the total negative difference $r^-$, defined similarly to the total positive difference:
\begin{definition}
    \label{def:negative-part}
    Define $r^-(t)$, the \textit{negative part} of the difference in work between the SMOD and SRPT-$k$ systems,
    to be
    \begin{align*}
        r^-(t) := \sum_{i} |\min(b_i^{SMOD}(t) - b_i^{SRPT\hy k}(t), 0)|.
    \end{align*}
\end{definition}

This minor lemma covers the integrated difference after SMOD dominates SRPT-$k$.
\begin{lemma}
    \label{lem:negative-bounds-diff}
    If at time $t_1$, SMOD dominates SRPT-$k$, then for any times $t_2 > t_1$,
    the IND is lower bounded:
    \begin{align*}
        &\int_{u=t_1}^{t_2} N^{SRPT\hy k}(u) - N^{SMOD}(u) du
        \ge k(r^-(t_1) - r^-(t_2)).
    \end{align*}
\end{lemma}
\begin{proof}{Proof of \cref{lem:negative-bounds-diff}}
    Note that arrivals and completions do not change $r^-$.
    Note also that if the same number of jobs are served in both the SMOD and SRPT-$k$ systems,
    $r^-$ also does not change.
    Thus whenever $r^-$ decreases, \new{fewer} jobs are being served in the SMOD system than in the SRPT-$k$ system.
    As a result, \new{fewer} jobs must be present in the SMOD system than in the SRPT-$k$ system.
    In particular, the rate of decrease of $r^-$ is lower bounded by the difference in number of jobs in the two systems:
    \begin{align*}
        \frac{d}{dt}r^-(t) \ge \frac{N^{SMOD}(t) - N^{SRPT\hy k}(t)}{k}.
    \end{align*}
    Rearranging and integrating completes the proof. 
    \end{proof}

\subsubsection{Work inequality \& positive-less-than-negative}
\label{sec:work-ineq}

We now show a form of work monotonicity:
Under conditions that we specify below, the SMOD system always has at most as much work present as the SRPT-$k$ system, that is, $W^{SMOD}\leq W^{SRPT\hy k}$.  
The preservation of $W^{SMOD} \le W^{SRPT\hy k}$ is key to our result, and the core of the success of the SMOD policy.
To prove this work monotonicity property, we first define another preserved property.

\begin{definition}
    An index $i$ is called a \textit{negative-diff index} at time $t$ if $b^{SMOD}_i(t) < b^{SRPT\hy k}_i(t)$. Similarly, $i$ is a \textit{positive-diff index} at time $t$ if $b^{SMOD}_i(t) > b^{SRPT\hy k}_i(t)$.
\end{definition}

Now we can define the PLN property:
\begin{definition}
    \label{def:pln}
    The SMOD and SRPT-$k$ joint system has the \textit{positive-less-than-negative (PLN)} property
    if every positive-diff index is smaller than every negative-diff index.
\end{definition}

Note as a special case that if there are no negative-diff indices, PLN holds,
and if there are no positive-diff indices, PLN holds. 
Intuitively, a positive-diff index $i$ indicates that the $i$-th job has a larger remaining size in the SMOD system, and negative-diff implies that the corresponding job has \new{smaller} remaining size in the SMOD system. Hence, the PLN property says that \new{the} SMOD system has more work concentrated in the smallest indices compared to the SRPT-$k$ system. Note in particular that if SMOD dominates SRPT-$k$, there are no positive-diff indices, and PLN holds.

Note that if the PLN property holds, then either $N^{SMOD}(t) \ge N^{SRPT\hy k}(t)$, or SMOD dominates SRPT-$k$.
To see why, note that if a positive-diff index exists, then index 1 cannot be a negative-diff index, and specifically cannot be a padding index in the SMOD system.
If there is no padding index, then $N^{SMOD} \ge N^{SRPT\hy k}$.
If no positive-diff index exists, this is the definition of dominance.

We now show that if all of PLN, $W^{SMOD} \le W^{SRPT\hy k}$ and zig-zag matching hold at some point in time, this joint property is preserved.

\begin{lemma}
    \label{lem:pln}
    If at time $t_1$, the systems have the PLN property and $W^{SMOD}(t_1) \le W^{SRPT\hy k}(t_1)$ and zig-zag matching,
    then at all times $t_2 \ge t_1$, the systems continue to have the same relationship:
    the PLN property, $W^{SMOD}(t_2) \le W^{SRPT\hy k}(t_2)$ and zig-zag matching all hold.
\end{lemma}
\begin{proof}{Proof of \cref{lem:pln}}
    Note that zig-zag matching is preserved by \cref{lem:zig-zag}, so we focus on establishing the other two properties, using zig-zag matching for the proof.

    We know $W^{SMOD}(t_1) \le W^{SRPT\hy k}(t_1)$, which is equivalent to
    $\sum_i b_i^{SMOD}(t_1) \leq \sum_i b_i^{SRPT\hy k}(t_1).$
    With some algebraic manipulation, this inequality is equivalent to
    \begin{align*}
        &\sum_i (b_i^{SMOD}(t_1)-b_i^{SRPT\hy k}(t_1))^+ 
        \leq \sum_j (b_j^{SRPT\hy k}(t_1)-b_j^{\new{SMOD}}(t_1))^+.
    \end{align*}
    The left-hand side summation terms are nonzero for the positive-diff indices, and the right-hand side, the negative-diff indices. 

    If there is no positive-diff index at time $t_1$, SMOD dominates SRPT-$k$ at time $t_1$, and by \cref{cor:dominance}, dominance is preserved at every time $t_2\geq t_1$. 
    Note that dominance implies both the PLN property and $W^{SMOD} \le W^{SRPT\hy k}$,
    so if there is no positive-diff index at time $t_1$, the lemma must hold.
    
    On the other hand, if there is a positive-diff index, the PLN property tells us that the first index is not a negative-diff index. Therefore, we know that $b^{SMOD}_1(t_1) > 0$. If $b^{SRPT\hy k}_1(t_1) = 0$, then a padding index exists and $N^{SMOD}(t_1) > N^{SRPT\hy k}(t_1)$. If $b^{SRPT\hy k}_1(t_1) > 0$, then $N^{SMOD}(t_1) = N^{SRPT\hy k}(t_1)$, i.e., no padding. Hence, $N^{SMOD}(t_1) \ge N^{SRPT\hy k}(t_1)$.

    Additionally, the zig-zag property holds at time $t_1$, and \cref{lem:zig-zag} implies $N^{SMOD}(t_1) \le N^{SRPT\hy k}(t_1) + 1$. Hence, there are two possible values for $N^{SMOD}: N^{SRPT\hy k}$ and $N^{SRPT\hy k} + 1$.
    As a result, SMOD will always serve at least as many jobs as SRPT-$k$, so $W^{SMOD} \le W^{SRPT\hy k}$ is preserved if PLN is preserved. The proof that PLN is preserved follows by careful analysis of the effects on the indices in the SMOD and SRPT-$k$ systems of the following events: arrivals, continuous service, completions and a difference reaching zero. The approach is similar to the proof of \cref{lem:pos-part} and we present it in \cref{ap:proof-pln}. 
    \end{proof}

\subsubsection{Time of dominance and maximum size job}
\label{sec:time-of-dominance}

We bound the time at which SMOD achieves dominance relative to SRPT-$k$, relative to the completion time of the initially-largest job  in the SRPT-$k$ system. \new{We start with a definition and then prove the result.}

\begin{definition}    
At a given time $t$, we consider the \textit{maximum index $i_*(t)$} for which $b_i^{SMOD}(t) \neq b_i^{SRPT\hy k}(t)$.
If there is no such index, we consider $i_*(t)$ to be undefined.
We tag the job at index $i_*(t)$ in the SRPT-$k$ system as the \textit{max-$t$ job} and denote it by $m_t$.
Note that this job's index in the SRPT-$k$ queue will change throughout its time in that system.
Let $m_t(t_1)$ denote the max-$t$ job's index in the SRPT-$k$ queue at time $t_1 \ge t$.
\end{definition}

With this definition, we can state our final major lemma, which bounds the duration until the SMOD system dominates the SRPT-$k$ system.

\begin{lemma}
    \label{lem:max-diff}
    Suppose that at time $t_1$ the systems obey the PLN property, zig-zag matching,
    and that $W^{SMOD}(t_1) \le W^{SRPT\hy k}(t_1)$.
    Then, at the time the max-$t_1$ job $m_{t_1}$ completes in the SRPT-$k$ system,
    the SMOD system dominates the SRPT-$k$ system.
\end{lemma}
\begin{proof}{Proof of \cref{lem:max-diff}}
    Let $t_{comp}$ denote the time at which the max-$t_1$ job $m_{t_1}$ completes in the SRPT-$k$ system. We show that at all times $t_2 \in [t_1, t_{comp}]$,
    if the SMOD system does not dominate the SRPT-$k$ system,
    the largest index $i_*(t_2)$ for which the two systems differ is bounded by
    $m_{t_1}(t_2)$, the index of the max-$t_1$ job at time $t_2$.
    Thus, when the max-$t_1$ job completes at time $t_{comp}$,
    either the SMOD system dominates the SRPT-$k$ system,
    or there is no index where the systems differ. 
    It suffices to show that at all times $t_2 \in [t_1, t_{comp}]$, either SMOD dominates SRPT-$k$
    or $i_*(t_2) \le m_{t_1}(t_2)$.
    
    First, note that by \cref{lem:pln},
    at all times $t_2 \in [t_1, t_{comp}]$, either SMOD dominates SRPT-$k$, or
    the systems continue to satisfy PLN, $W^{SMOD}(t_2)\le W^{SRPT\hy k}(t_2)$ and zig-zag matching.
    
    If SMOD does not dominate SRPT-$k$ at time $t_2$,
    i.e., if a positive-diff index exists,
    then because $W^{SMOD}(t_2) \le W^{SRPT\hy k}(t_2)$,
    a negative-diff index must exist.
    By the PLN property
    we know that the index $i_*(t_2)$
    must be a negative-diff index. Then, we assume that at all times $t$ prior to $t_2$, we have $i_*(t) \le m_{t_1}(t)$,
    and we seek to demonstrate that the property is maintained. The rest of the proof follows immediately after considering the effects of the following events in the system: jump change on arrival, jump change on a completion, continuous service, and a diff reaching zero. We provide the details in \cref{ap:proof-max-diff}. 
    \end{proof}

\subsection{SEK-SMOD-specific worst-case results}
\label{sec:sek-smod-wc}

Using the SMOD monotonicity and preservation properties proven in \cref{sec:smod-wc}, we prove SEK-SMOD properties in the following lemma. \new{This lemma connects all the properties we have proved for SEK-SMOD, and it is essential for the proof of \cref{thm:main}.}

\begin{lemma}
    \label{lem:sek-bridge}
    At all times, the SEK-SMOD scheduling policy maintains the following properties relative
    to the coupled SRPT-$k$ system:
    \begin{description}
        \item[Work inequality:] The SEK-SMOD system always has at most as much work present as the SRPT-$k$ system: $W^{SEK \hy SMOD}(t) \le W^{SRPT\hy k}(t)$.
        \item[Zig-zag:] The SEK-SMOD system always exhibits zig-zag matching, defined in \cref{def:zig-zag}, with respect to the SRPT-$k$ system.
        \item[PLN:] The SEK-SMOD system always exhibits the positive-less-than-negative property, defined in \cref{def:pln}, with respect to the SRPT-$k$ system.
        \item[Positive part bound:] The SEK-SMOD system always has a bounded positive part $r^+ \le \epsilon$ with respect to the SRPT-$k$ system.
    \end{description}
\end{lemma}

\begin{proof}{Proof of \cref{lem:sek-bridge}}
    In \cref{lem:pln,lem:pos-part}, we showed that the SMOD policy maintains the work inequality, the zig-zag property, the PLN property,
    and has monotonically decreasing positive part $r^+$.
    When the two systems are merged, all four properties hold trivially.
    Thus, it suffices to show that all four properties hold throughout the SEK period.

    To demonstrate these properties, we refer to the definition of the SEK policy in \cref{def:sek}.
    Note that the $k-1$ smallest jobs at the divergence point received the same service in both systems,
    so we may focus only on the largest two jobs at the divergence point.
    Let $b_1, b_k, b_{k+1}$ be the remaining sizes of the smallest, second-largest, and largest jobs
    at the divergence point.    
    After $t$ time since divergence, these two largest jobs' remaining sizes are $[b_k - (t - kb_1)^+/k, b_{k+1} - t/k]$
    under SEK, and $[b_k - t/k, b_{k+1} - (t - kb_1)^+/k]$ under SRPT-$k$.
    We use these remaining sizes to demonstrate the four properties.
    Note that either of $b_k - (t - kb_1)^+/k$ and $b_{k+1} - t/k$
    may have the largest remaining size at a given time $t$. We cover both possibilities.
    \begin{description}
        \item[Work inequality:] We show that, until the time when the SEK policy finishes its final job,
        SEK always serves at least as many jobs as SRPT-$k$, implying the work inequality.
        This is equivalent to showing that the $k$th completion time under SEK,
        the first of the two largest jobs to complete,
        is no sooner under SEK than under SRPT-$k$.
        In the SEK system, the second-largest job finishes at time $t = k(b_1 + b_k)$, and the largest job finishes at time $kb_{k+1}$. In the SRPT-$k$ system, the second-largest job finishes at time $kb_k$, and the largest job finishes at time $k(b_1 + b_{k+1})$.
        Note that $b_1 + b_k \ge b_k$, and $b_{k+1} \ge b_k$,
        so these final two jobs under SEK finish after the second-to-last completion under SRPT-$k$.
        This confirms the work inequality.
        \item[Zig-zag:] We must ensure that neither of the two largest jobs' remaining sizes under SEK
        exceed that of the largest job under SRPT-$k$.
        But note that $b_k - (t - kb_1)^+/k \le b_{k+1} - (t - kb_1)^+/k$,
        and $b_{k+1} - t/k \le b_{k+1} - (t - kb_1)^+/k$, confirming the zig-zag property.
        \item[PLN:] The only indices that may have nonzero difference are those corresponding to the two largest jobs. The second-largest job has a positive difference (larger under SEK) and the largest job has a negative difference (smaller under SEK). Thus, the PLN property holds.
        \item[$r^+$ at most $\epsilon$:] The only positive difference index is that of the second-largest job.
        This difference is at most $(b_k - (t - kb_1)^+/k) - (b_k - t/k) \le b_1 \le \epsilon$.
        
    \end{description}\end{proof}

\section{Numerical study}
\label{sec:empirical}

We have proven that SEK-SMOD outperforms SRPT-$k$ at all loads, with appropriate parameterization.
In this section, we numerically study the Practical SEK policy, defined in \cref{def:practical-sek}\new{, a simplified single-parameter version of the SEK policy defined in \cref{def:sek}}.
We demonstrate that Practical SEK achieves significant, reliable improvements over SRPT-$k$, across a wide range of distributions \new{$F_S$},  loads $\rho$, and numbers of servers $k \ge 2$. In \cref{sec:empirical-old-policies} we compare Practical SEK with the three scheduling policies which had previously been shown to achieve optimal mean response time in heavy traffic: SRPT, PSJF (Preemptive Shortest Job First) and RS (Remaining Size times Original Size) \citep{grosof-etal-2018-srpt-k}.
We start with three baseline size distributions in \cref{sec:empirical-old-policies}, and explore a variety of higher-variance distributions in \cref{sec:empirical-explore-dists}, where Practical SEK's benefit is empirically largest compared to SRPT-$k$.

In \cref{sec:empirical-more-servers} we explore settings with more servers, $k \in \{2, 3, \ldots, 6\}$. We show that the exponentially-diminishing improvement lower bound in \cref{lem:any-diff}, with parameters defined in \cref{def:parameter}, is a proof artifact: Practical SEK's improvement over SRPT-$k$ remains significant even as $k$ increases, decreasing roughly linearly with $k$, not exponentially. We also introduce the SEK-$n$ variant of the Practical SEK policy, which is better suited to settings with more servers, defined in \cref{def:sek-n} and explored in \cref{sec:empirical-more-servers}.

Finally, in \cref{sec:empirical-estimates}, we show that Practical SEK's improvement over SRPT-$k$ remains significant even if the scheduler only has access to estimates of job size, not exact job size information, though estimates must be somewhat accurate for SEK to be useful. \new{We discuss partially-known- and unknown-size settings further in \cref{sec:discuss-size}.}

Before presenting the results, we properly define the Practical SEK policy.
\begin{definition}\label{def:practical-sek}
    Practical SEK has a single parameter $\epsilon > 0$.
    Practical SEK serves the $k$ jobs of least remaining size (matching SRPT-$k$), except when the following scenario arises:
    \begin{enumerate}
        \item There are exactly $k+1$ jobs present,
        \item $k$ of the jobs have remaining size less than $\epsilon$, and
        \item the largest job has remaining size greater than $\epsilon$.
    \end{enumerate}
    In this case, Practical SEK serves the $k-1$ jobs of least remaining size, and the final job of largest remaining size
    until an arrival or completion occurs.
\end{definition}

For larger numbers of servers $k \ge 3$, we define the SEK-$n$ policy, which generalizes the situations in which jobs can be reordered.
We numerically explore SEK-$n$ in \cref{sec:empirical-more-servers}.
\begin{definition}\label{def:sek-n}
    SEK-$n$ has two parameters, $\epsilon > 0$ and $n \ge 1$. SEK-1 matches Practical SEK defined in \cref{def:practical-sek}.
    For general $n$, SEK-$n$ serves the $k$ jobs of least remaining size (matching SRPT-$k$), except when the following scenario arises:
    \begin{enumerate}
        \item There are between $k+1$ and $k+n$ jobs present,
        \item $k$ of the jobs have remaining size less than $\epsilon$, and
        \item all other jobs have remaining size greater than $\epsilon$.
    \end{enumerate}
    In this case, SEK-$n$ serves the $k-1$ jobs of least remaining size, and the job of $(k+1)$st least remaining size
    until an arrival or completion occurs, skipping over the job of $k$th least remaining size.
\end{definition}

\subsection{Comparison with SRPT, PSJF, and RS}
\label{sec:empirical-old-policies}

\begin{figure}
{\subfloat[$S \sim Exp(1)$.\label{fig:exp-2}]{\includegraphics[width=0.32\linewidth]{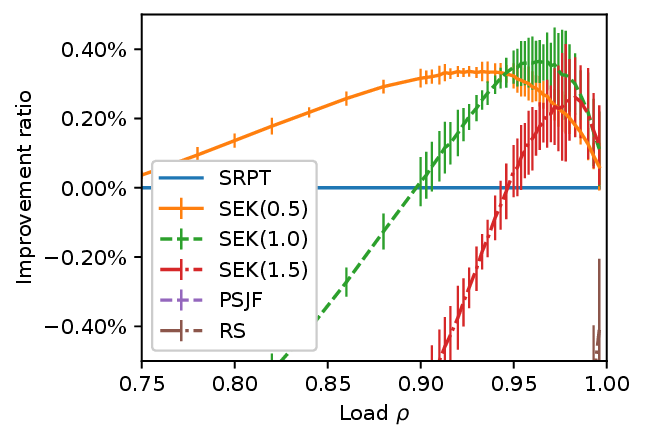}} }
{\subfloat[$S\sim \text{HypExp}, \E{S}=1, C^2\approx 10$.\label{fig:hyperexp-2}]{\includegraphics[width=0.32\linewidth]{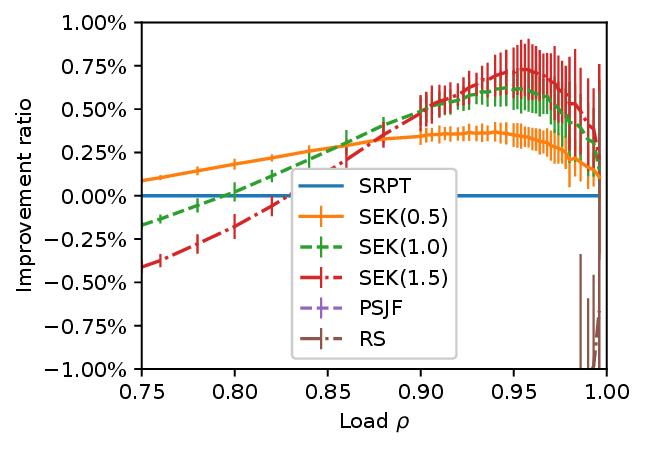}} }
{\subfloat[$S \sim Uniform(0, 2)$.\label{fig:uniform-2}]{\includegraphics[width=0.32\linewidth]{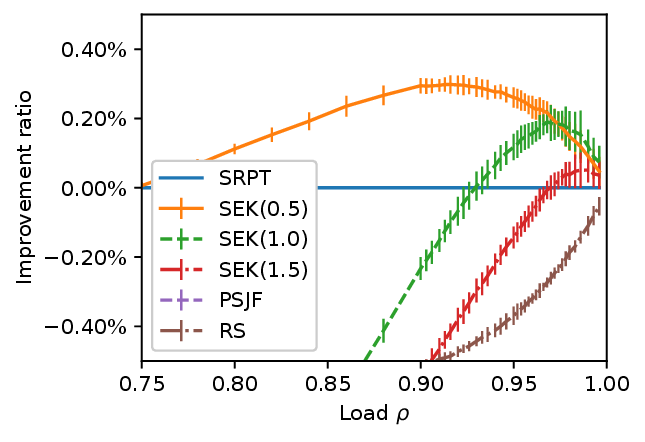}}}
\caption{Improvement of SEK over various policies.
SEK vs. SRPT-$k$, PSJF-$k$, and RS-$k$, under varying distributions \new{for} $S$ and threshold parameters $\epsilon \in \{0.5, 1, 1.5\}$. Improvement ratio compares each policy to SRPT-$k$. Note that PSJF-$k$ and RS-$k$ are off the negative side of the chart when not visible. \new{Vertical bars are 95\% confidence intervals.}  Loads simulated: $\rho \in [0.75, 0.996]$, $k=2$ servers. \new{$10$ independent runs of $10^6$} arrivals per data point.}
\label{fig:old-zoom-in}
\end{figure}

In \cref{fig:old-zoom-in}, we examine three values of $\epsilon$ parameter: $0.5, 1,$ and $1.5$, for three size distributions with mean $\E{S} = 1$: An exponential distribution in \new{Fig.~\hyperref[fig:exp-2]{\ref{fig:old-zoom-in}a}}, 
a hyperexponential distribution with squared coefficient of variation $C^2 \approx 10$ in \new{Fig.~\hyperref[fig:hyperexp-2]{\ref{fig:old-zoom-in}b}},
and a uniform distribution in 
\new{Fig.~\hyperref[fig:uniform-2]{\ref{fig:old-zoom-in}c}}.
In this section, we focus on the $k=2$ server, perfect information case. See \cref{sec:empirical-more-servers} for $k>2$ servers, and \cref{sec:empirical-estimates} for approximate size information.


In these figures, we plot the empirical ``improvement ratio'' metric $1 - \frac{\E{T^\pi}}{\E{T^{SRPT\hy k}}}$, where $\pi$ represents the SEK policy for the three values of $\epsilon$ specified above, or the SRPT-$k$, PSJF-$k$, and RS-$k$ policies as baselines. Note that values above 0 indicate improvement upon SRPT-$k$, and values below 0 indicate worse performance than SRPT-$k$.

Note that RS-$k$ never outperforms SRPT-$k$, and is barely visible in 
\new{Fig.~\hyperref[fig:exp-2]{\ref{fig:old-zoom-in}a}} 
and 
\new{Fig.~\hyperref[fig:hyperexp-2]{\ref{fig:old-zoom-in}b}}, 
while PSJF-$k$ is not visible in any of the plots. To make PSJF-$k$ visible, we must expand the vertical axis by a factor of 15-30$\times$, as shown in \cref{fig:old-zoom-out}.

Returning to \cref{fig:old-zoom-in} and examining the different threshold parameters $\epsilon$ for the Practical SEK policy, we see that $\epsilon=0.5$ achieves some improvement over SRPT-$k$ across the entire range of loads simulated, with a low sensitivity to job size distribution variability, while the larger thresholds $\epsilon=1$ and $\epsilon=1.5$ achieve higher peak improvement ratios \new{for the higher-variance job size distributions}, with a peak improvement ratio of $0.32\%$ for $\epsilon=1$ in the exponential case in 
\new{Fig.~\hyperref[fig:exp-2]{\ref{fig:old-zoom-in}a}}, 
$0.67\%$ for $\epsilon=1.5$ in the hyperexponential case in 
\new{Fig.~\hyperref[fig:hyperexp-2]{\ref{fig:old-zoom-in}b}}, 
and $0.20\%$  \new{for $\epsilon = 0.5$} in the uniform case in 
\new{Fig.~\hyperref[fig:uniform-2]{\ref{fig:old-zoom-in}c}}.

The results for these distributions are indicative of a wider trend: Practical SEK is more beneficial, relative to SRPT-$k$, when the job size distribution has higher variance. We explore this trend further in \cref{sec:empirical-explore-dists}.

\begin{figure}
{\subfloat[$S \sim Exp(1)$.\label{fig:exp-2-out}]{\includegraphics[width=0.32\linewidth]{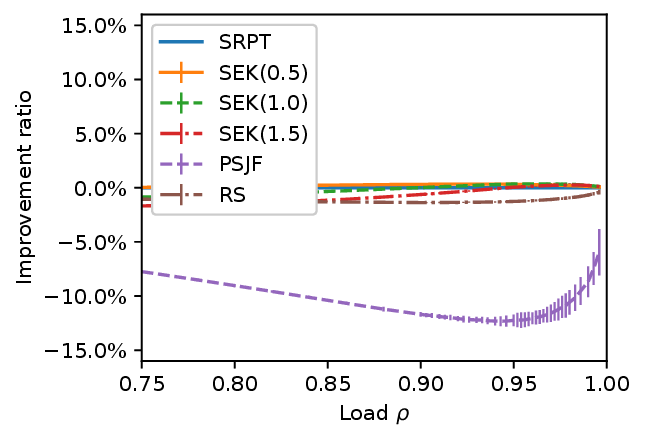}} }
{\subfloat[$S\sim \text{HypExp}, \E{S}=1, C^2\approx 10$.]{\includegraphics[width=0.32\linewidth]{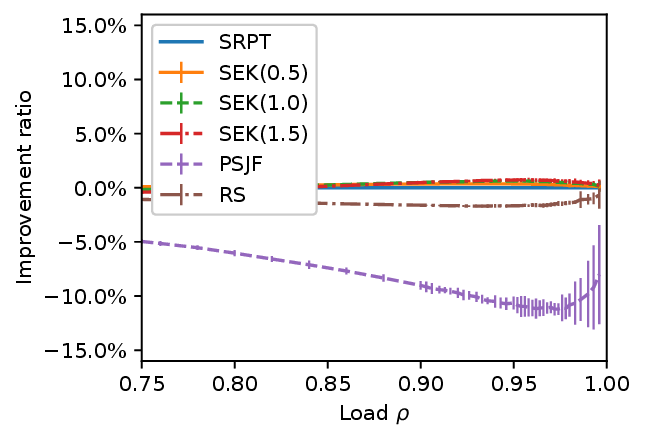}}}
{\subfloat[$S \sim Uniform(0, 2)$.\label{fig:uniform-2-out}]{\includegraphics[width=0.32\linewidth]{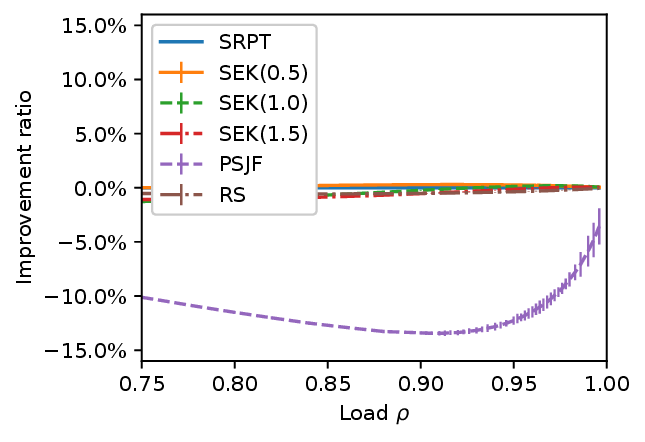}}}

\caption{Improvement of SEK over various policies with an expanded vertical axis. SEK vs. SRPT-$k$, PSJF-$k$, and RS-$k$, under varying distributions \new{for} $S$ and threshold parameters $\epsilon \in \{0.5, 1, 1.5\}$, with an expanded vertical axis. 
Improvement ratio compares each policy to SRPT-$k$. 
\new{Vertical bars are 95\% confidence intervals.}
Loads simulated: $\rho \in [0.75, 0.996]$, $k=2$ servers. \new{$10$ independent runs of $10^6$} arrivals per data point.}
\label{fig:old-zoom-out}
\end{figure}

\subsection{Systematic exploration of higher-variance distributions}
\label{sec:empirical-explore-dists}

In this section, we show that SEK demonstrates larger improvement over SRPT-$k$, under higher variance job size distributions.
We specifically focus on 2-branch hyperexponential job size distributions, with $k=2$ servers. For $k>2$ servers, see \cref{sec:empirical-more-servers}.

We parameterize our size distributions by squared coefficient of variability, $C^2 = \frac{\text{Var}(S)}{\E{S}^2}$, and by the fraction of load $\rho_{high}$ accounted for by the high-mean branch of the hyperexponential distribution. We normalize the mean $\E{S}=1$.
In \cref{tbl:explore-dists}, we explore job size distributions defined by all pairs of $C^2 \in \{2, 4, 10, 20, 40\}$
and high-mean load fraction $\in \{0.1, 0.3, 0.5, 0.7, 0.9\}$,
across loads $\rho \in \{0.9, 0.91, 0.92, 0.93, 0.94, 0.95, 0.96, 0.97, 0.98, 0.99\}$,
and across SEK thresholds $\epsilon \in \{1, 1.5, 2, 3\}$.
For each distribution, we give the $(\rho, \epsilon)$ pair which maximizes the improvement ratio $1 - \frac{\E{T^\pi}}{\E{T^{SRPT\hy k}}}$,
with the globally highest improvement ratio in bold.
The highest improvement ratio we found was a $0.92\%$ improvement, with $C^2=40$ and $\rho_{high}=0.5$, with $\rho=0.96$ and $\epsilon=3$.
Further exploration would likely uncover even larger improvement ratios.

\begin{table}
 \mbox{}
 \hspace{-30pt}
\begin{minipage}{0.48\textwidth}
\begin{tabular}{c|c||c|c|c|c}
\hspace{-4pt}$C^2$ \hspace{-4pt} & \hspace{-4pt} $\rho_{high}$ \hspace{-4pt} & \hspace{-4pt} Best Improve \hspace{-4pt} & \hspace{-4pt} Best $\rho$ \hspace{-4pt} & \hspace{-4pt} Best $\epsilon$ \hspace{-4pt}& \new{$\sigma$} \\
\hline
2& 0.1& 0.5162\%& 0.98& 1 & 0.0873\% \\
2& 0.3& 0.5002\%& 0.97& 1 & 0.0422\% \\
2& 0.5& 0.5112\%& 0.96& 1 & 0.0231\% \\
2& 0.7& 0.4847\%& 0.96& 1 & 0.0374\% \\
2& 0.9& 0.4458\%& 0.94& 1 & 0.0333\% \\
4& 0.1& 0.5926\%& 0.98& 1 & 0.1124\% \\
4& 0.3& 0.6223\%& 0.97& 1.5 & 0.0800\% \\
4& 0.5& 0.6205\%& 0.96& 1.5 & 0.0510\% \\
4& 0.7& 0.6074\%& 0.95& 1.5 & 0.0528\% \\
4& 0.9& 0.5154\%& 0.95& 2 & 0.0492\% \\
10& 0.1& 0.6019\%& 0.98& 1 & 0.1262\% \\
10& 0.3& 0.7622\%& 0.98& 3 & 0.1643\% \\
10& 0.5& 0.7966\%& 0.96& 3 & 0.1247\% \\
\end{tabular}
\end{minipage}
\hfill
\begin{minipage}{0.48\textwidth}
\mbox{}
\hspace{-10pt}
\begin{tabular}{c|c||c|c|c|c}
\hspace{-4pt}$C^2$ \hspace{-4pt} & \hspace{-4pt} $\rho_{high}$ \hspace{-4pt} & \hspace{-4pt} Best Improve \hspace{-4pt} & \hspace{-4pt} Best $\rho$ \hspace{-4pt} & \hspace{-4pt} Best $\epsilon$ \hspace{-4pt}& \new{$\sigma$}\\
\hline
10& 0.7& 0.7336\%& 0.93& 3 & 0.0333\% \\
10& 0.9& 0.5131\%& 0.93& 3 & 0.0304\% \\
20& 0.1& 0.6622\%& 0.98& 1 & 0.0974\% \\
20& 0.3& 0.8902\%& 0.98& 3 & 0.2096\% \\
20& 0.5& 0.9008\%& 0.96& 3 & 0.1283\% \\
20& 0.7& 0.6828\%& 0.92& 3 & 0.0415\% \\
20& 0.9& 0.3748\%& 0.91& 3 & 0.0224\% \\
40& 0.1& 0.6357\%& 0.97& 1 & 0.1878\% \\
40& 0.3& 0.8554\%& 0.97& 3 & 0.2192\% \\
40& 0.5& \textbf{0.9163\%}& 0.96& 3 & 0.1253\% \\
40& 0.7& 0.6347\%& 0.92& 3 & 0.0477\% \\
40& 0.9& 0.2207\%& 0.9& 3 & 0.0267\% \\
    \end{tabular}
\end{minipage}

\caption{Peak improvement of SEK over SRPT-$k$. Experiments varying 2-branch hyperexponential distributions $S$, maximized over load $\rho$ and threshold parameter $\epsilon$. $k=2$ servers. \new{$\sigma$~specifies standard deviation of improvement ratio for specified parameters.
$10$ independent runs of $10^6$} arrivals per data point.}
 \label{tbl:explore-dists}
\end{table}

Consistent trends are visible in \cref{tbl:explore-dists}:
Larger improvements are associated with higher $C^2$ and lower $\rho_{high}$, and the best threshold $\epsilon$ increases with larger $C^2$ and larger $\rho_{high}$.
\new{We discuss the settings in which these improvements are most likely to be practically meaningful in \cref{sec:discuss-meaningful}.}

\subsection{Larger number of servers k, and policy variant SEK-n}
\label{sec:empirical-more-servers}

\begin{figure}
{
{\subfloat[Practical SEK policy.\label{fig:multi-sek-1}]{\includegraphics[width=0.3\linewidth]{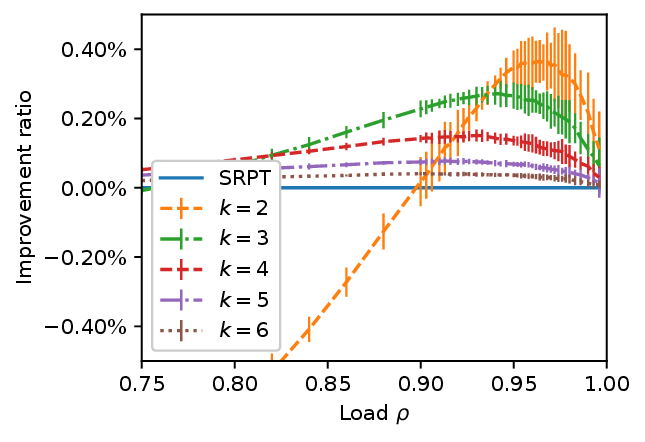}}}
{\subfloat[SEK-2 policy.\label{fig:multi-sek-2}]{\includegraphics[width=0.3\linewidth]{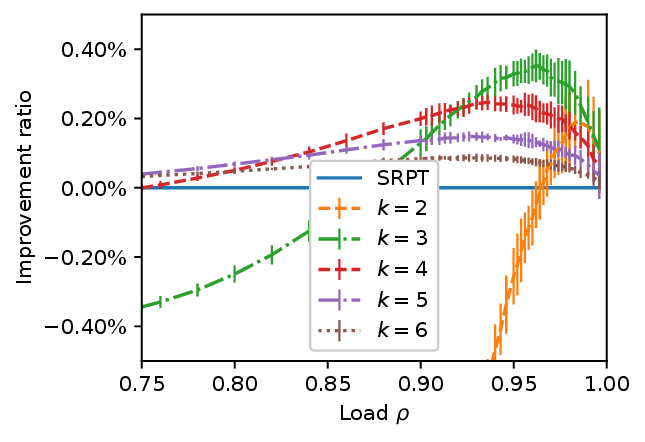}}}
{\subfloat[SEK-3 policy.\label{fig:multi-sek-3}]{\includegraphics[width=0.3\linewidth]{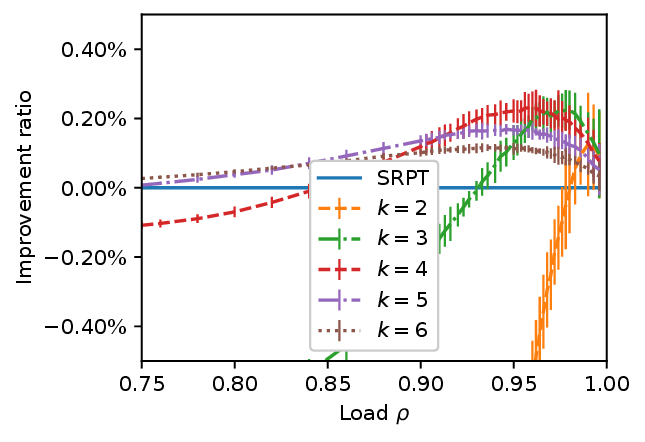}}}
}
\caption{Improvement of SEK over SRPT-$k$, and the SEK-2 and SEK-3 variants.
Numerical analysis for $k=2, 3, \ldots, 6$ servers. Threshold $\epsilon=1$.
\new{Vertical bars are 95\% confidence intervals.} Loads simulated: $\rho \in [0.75, 0.996]$. Job size: $S \sim Exp(1)$. \new{$10$ runs of $10^6$} arrivals per data point.}
\label{fig:multi-sek}
\end{figure}

In this section, we examine the improvement of the SEK policy over the SRPT-$k$ policy for $k=2, 3, \ldots, 6$ servers.
In \new{Fig.~\hyperref[fig:multi-sek-1]{\ref{fig:multi-sek}a}}, we compare the standard Practical SEK policy with SRPT-$k$ over this range of servers.
Here we see that for $k \ge 3$ servers, improvement over SRPT-$k$ is more consistent across a range of loads, extending to all loads simulated.
However, peak improvement is smaller as $k$ increases. Note however that the exponential decrease in $k$ of the improvement lower bound proven in \cref{lem:any-diff}, with parameters given in \cref{def:parameter}, is overly pessimistic:
Even at $k=5$ \new{or $6$} servers, improvement is still significant \new{and visible}.

To enlarge peak improvement for more servers $k$, we \new{explore variants} of the Practical SEK policy,
which we call the SEK-$n$ family of policies, which we defined in \cref{def:sek-n}.
Under Practical SEK, we serve the $k+1$st smallest job in front of the $k$th smallest job if and only if there are $k+1$ jobs in the system,
of which $k$ have remaining size below $\epsilon$ and $1$ has remaining size above $\epsilon$.
Under SEK-$n$, we generalize the cases in which we serve the $k+1$st smallest job in front of the $k$th smallest job
to allow up to $n$ jobs with remaining size above $\epsilon$, while still requiring exactly $k$ jobs with remaining size below $\epsilon$ to be present.

The intuition behind the SEK-$n$ policy is similar to that of the baseline SEK policy:
We find situations where the SRPT-$k$ policy would waste capacity by prematurely emptying one or more servers by completing the smallest jobs,
and instead serve a larger job to more efficiently utilize capacity. Note that this intuition only applies when $n < k$, and we see in \cref{fig:multi-sek}
that SEK-$n$ is most beneficial precisely when this is satisfied.

In particular, in \new{Fig.~\hyperref[fig:multi-sek-2]{\ref{fig:multi-sek}b}} and \new{Fig.~\hyperref[fig:multi-sek-3]{\ref{fig:multi-sek}c}}, we show the improvement of the SEK-2 and SEK-3 policies over the SRPT-$k$ policy, again for $k=2,3,\ldots,6$ servers. For $k=3$, the peak improvement is largest under SEK-2, with similar peak improvement ratio as for SEK-1 with $k=2$ servers. $k=4$ similarly shows largest improvement under SEK-2, while for $k=5$ and $k=6$, \new{SEK-3 achieves} the largest improvement.

\subsection{Estimated job sizes}
\label{sec:empirical-estimates}

Previously, we have focused on the case where the scheduler has access to exact size information.
In this section, we consider the case where the scheduler only has access to estimated size information.

In this setting, when each job arrives, the scheduler has access to a size estimate $E := (S + D)^+$,
where $D$ is an i.i.d. error margin.
As our baseline policy, we consider the SRPT-Estimate policy, which prioritizes jobs according to the remaining size estimate $E - a$,
where $a$ is the \new{job's} age (i.e., time in service).
Note that more refined or optimized \new{estimated-size policies} are known in the literature \citep{scully-et-al2022uniform,gittins1979bandit,gittins2011multi},
but this estimate makes a useful comparison point for testing SEK.
Specifically, we compare SRPT-Estimate against the SEK-Estimate policy,
which plugs the same estimate into \new{the} Practical SEK policy defined in \cref{def:practical-sek}.

\begin{figure}
    {
    \includegraphics[width=0.8\linewidth]{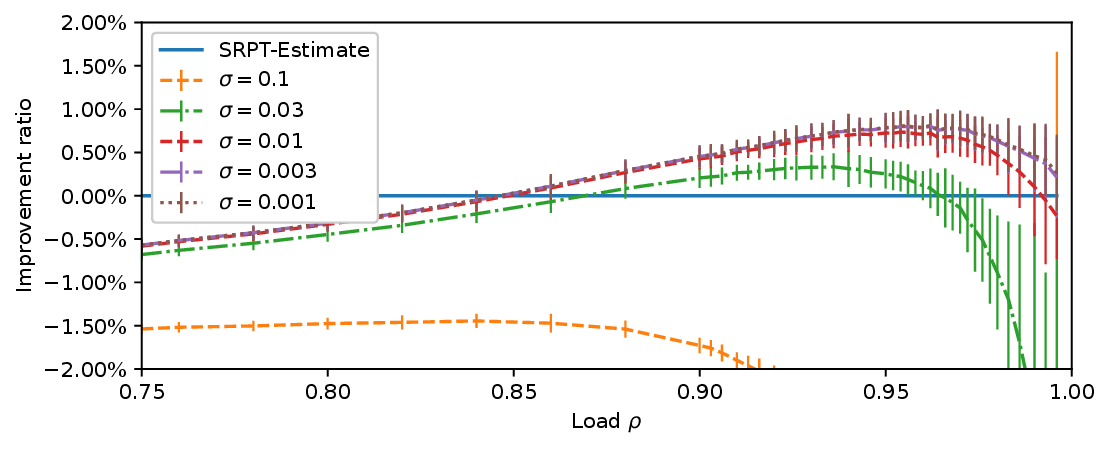}}
    \caption{Improvement of SEK-Estimate over SRPT-Estimate. 
    Experiments with normally distributed error $D$ with standard deviation $\sigma$.
    Job size distribution \new{$F_S$} is hyperexponential with $C^2 \approx 10$ and $\E{S}=1$.
    Threshold $\epsilon = 2$.
    \new{Vertical bars are 95\% confidence intervals.}
    Loads simulated: $\rho \in [0.75, 0.996]$, $k=2$ servers. \new{$10$ runs of $10^6$} arrivals per data point.}
    \label{fig:estimate}
\end{figure}

In \cref{fig:estimate}, we empirically demonstrate that SEK-Estimate can improve upon SRPT-Estimate,
for accurate enough estimates.
Specifically, we consider the error distribution $D \sim Normal(0, \sigma)$,
with the standard deviation $\sigma \in \{0.1, 0.03, 0.01, 0.003, 0.001\}$.
We use a choosier threshold $\epsilon=2$ than in \cref{sec:empirical-old-policies},
which we find produces better empirical results with estimated job sizes.

If errors are sufficiently large, e.g., $\sigma=0.1$, SEK-Estimate is worse than SRPT-Estimate.
However, as errors shrink, SEK-Estimate produces increasingly large improvements over SRPT-Estimate,
converging to the improvement margin in the perfect-information case.

Intuitively, the benefits of SEK over SRPT simply require SEK to be able to identify states in which near-future work efficiency is more important
than near-future completion speed. While sufficiently large estimation errors interfere with this identification
enough to make SEK-Estimate ineffective,
moderate to high accuracies of 3\% relative error or less are sufficient for SEK-Estimate to identify beneficial interchange states.
\new{We discuss considerations for larger-error or unknown-size settings in \cref{sec:discuss-size}.}

\section{\new{Discussion of practical considerations}}
\label{sec:discuss}

\new{We now explore practical considerations for the SEK policy, and specifically the Practical SEK policy defined in \cref{def:practical-sek}.
In \cref{sec:discuss-meaningful}, we discuss the practical settings in which SEK's improvement is most likely to be meaningful. In \cref{sec:discuss-implementation}, we discuss the implementation complexity and computational complexity of SEK. In \cref{sec:discuss-size}, we discuss the generalization of SEK to unknown-size and partially-known-size settings, building on our empirical evaluation in \cref{sec:empirical-estimates}.}

\subsection{\new{Meaningful improvement}}
\label{sec:discuss-meaningful}

\new{Our systematic exploration of Practical SEK in \cref{sec:empirical-explore-dists} showed a typical best improvement ratio of 0.5\%-1\% over SRPT-$k$, under high-variance job size distributions and moderately-high load $\rho$.
While this improvement ratio is small in absolute terms,
it is both theoretically and operationally meaningful.
Recent work has proven strong lower bounds on optimal mean response time in the M/G/$k$ system under any scheduling policy, showing that the improvement achieved by SEK represents a substantial fraction of the total distance to the optimal policy \citep{wang-et-al2025novel,izzy-MAMA-SEK}.

From an operational perspective, SEK carries no additional operational cost, no extra resources or labor, just improved performance. As a result, the performance improvement need only outweigh the one-time implementation cost to be operationally meaningful. This state of affairs is most common in large-scale, highly latency-sensitive environments, where a small implementation tweak can provide benefits across an entire fleet of service queues. For example, scheduling software is frequently deployed across thousands or millions of user devices, such as laptops or smartphones, where each device's two or four cores function similarly to a small M/G/$k$ system. Such large-scale, latency-sensitive settings represent the most advantageous practical use case for the SEK policy.}


\subsection{\new{Implementation \& computational complexity}}
\label{sec:discuss-implementation}

\new{The implementation complexity and computational complexity of Practical SEK slightly exceed that of SRPT-$k$. The most important difference lies in the fact that SRPT-$k$ only changes the set of jobs in service at arrival and completion moments, while SEK also does so at divergence points. However, divergence points are predictable in advance and can be calculated in just a few lines of code. Implementation can be handled using a timer interrupt or similar scheduling functionality. Computational overhead is likely to be quite low, as SEK uses the same information as SRPT-$k$ for decision making.}


\subsection{\new{Unknown-size \& partially-known-size settings}}
\label{sec:discuss-size}

\new{As shown in \cref{sec:empirical-estimates}, Practical SEK is still advantageous when high-quality job size estimates are available, not just exact size information.

When only lower-precision information is available, or even just the size distribution,
SRPT-$k$ is no longer the right baseline. Instead, the natural baseline policy becomes the Gittins-$k$ policy. The Gittins index policy achieves optimal mean response time in the M/G/1 system in the unknown- and partially-known-size settings \citep{gittins1979bandit,gittins2011multi,scully2021gittins}, paralleling SRPT.
Moreover, the Gittins-$k$ policy is known to achieve heavy-traffic optimal mean response time in the M/G/$k$ system \citep{scully-et-al2020gittins}, paralleling SRPT-$k$.

It is therefore natural to ask whether the ideas behind the SEK policy could be used to improve upon Gittins-$k$, with a Gittins-Except-$k+1$ (GEK) approach.
To devise such a GEK policy, one must replace the remaining size used by SEK
with a job's \emph{Gittins rank} or Gittins index. A job's Gittins rank is small if it has a substantial probability of completing after a short period of service.
Then, the GEK policy would diverge from Gittins-$k$ whenever $k$ jobs of low Gittins rank and one job of high Gittins rank are present.
Future research is needed to study whether such a GEK policy can consistently improve upon Gittins-$k$.}


\section{Conclusion}

We demonstrate the first scheduling policy which provably improves
on the mean response time of SRPT-$k$ in the M/G/$k$ queue.
We prove in \cref{thm:main} that our SRPT-Except-$k+1$ \& Modified SRPT (SEK-SMOD) policy has lower mean response time than SRPT-$k$ at all loads and under all job size distributions,
despite SRPT-$k$'s optimality results in the heavy-traffic limit.
These theoretical results are empirically validated in \cref{sec:empirical}.

To prove our results,
we combine hybrid worst-case-and-stochastic analysis with relative analysis.
We characterize the stochastic properties of the SEK portion of the policy in \cref{sec:stochastic}
and the worst-case \new{(i.e. sample-path)} properties of the SMOD portion of the policy in \cref{sec:worst-case},
both relative to the behavior of the baseline SRPT-$k$ policy.
We thereby demonstrate an explicit parameter choice for the SEK-SMOD policy
which always achieves lower mean response time than \new{the} SRPT-$k$ policy.

In \cref{sec:empirical}, we numerically explore the improvement of a practical variant of SEK over SRPT-$k$
under a variety of job size distributions \new{$F_S$}, numbers of servers $k$,
and under estimated job sizes, rather than perfect information.
We find that SEK's improvement over SRPT-$k$ is modest, but significant, peaking at around 0.9\% improvement,
which recent work indicates is a substantial fraction of the possible improvement from SRPT-$k$ to the optimal $k$-server policy \citep{wang-et-al2025novel}.
In \cref{sec:empirical-more-servers}, we also introduce a policy variant, SEK-$n$, which expands the situations in which a larger job is run ahead of a smaller job to improve performance for moderately larger numbers of servers $k \ge 3$.

An important direction for future work lies in optimization-based and learning-based search for additional improvement in mean response time over SRPT-$k$.
The continuous-time, continuous-state, infinite-dimensional nature of the known-size M/G/$k$ setting makes it hard to learn or optimize effective policies, but the lower-dimensional nature of SEK-type policies may be easier to search and optimize over.
Future work could also seek to improve relative to the Gittins-$k$ policy \citep{scully-et-al2020gittins}, \new{as we discuss in \cref{sec:discuss-size},}
allowing more complex partial information settings.


\section{Code and data disclosure}\label{sec:Code and Data Disclosure}
The code and generated data to support the numerical experiments are located in the \texttt{SEK-code-and-data.zip} file submitted alongside this manuscript as a supplemental file. This code and data will be released on GitHub in the non-anonymized version.


%
%
%

\bibliographystyle{ACM-Reference-Format}
\bibliography{references} 

\newpage

\begin{center}
    \new{\Large \textbf{Online Appendices}}
\end{center}

\appendix

\section{Summary of notation}
\label{ap:notation}

In this section, we review the notation used to define our model and scheduling policies, which is used repeatedly in the manuscript. The parameters and random variables associated with the M/G/$k$ system are presented in \cref{tab:notation-parameters}. These are mostly generic parameters, and are organized in order of appearance. 

In \cref{tab:notation-states}, we review the variables introduced to define \new{and analyze} the SEK-SMOD policy, including state descriptors and performance measures. These are organized alphabetically, with Greek letters listed at the end.

\renewcommand{\arraystretch}{1.2}
\begin{table}
  {\begin{tabular}{cp{0.85\textwidth}}
    \hline
    Notation & Meaning \\
    \hline
    $\lambda$ & Mean arrival rate. \\
    $S$ & Generally distributed random variable that represents the size of an arriving job. \\
    \new{$F_S$} & \leavevmode\new{Cumulative distribution function (cdf) of the job sizes upon arrival.} \\
    $k$ & Number of servers. Each of them processes jobs at a fixed rate of $1/k$. \\
    $\rho$ & Mean load of the system, that is, $\rho=\lambda\E{S}$. \\
    \hline
  \end{tabular}}
\caption{Notation for random variables and system parameters.}
\label{tab:notation-parameters}
\end{table}

\renewcommand{\arraystretch}{1.2}
\begin{table}
  {\begin{tabular}{cp{0.75\textwidth}}
    \hline
    Notation & Meaning \\
    \hline
    $\bb^A(t)$ & Vector of ordered remaining sizes at time $t$ in an M/G/$k$ queue under scheduling policy $A$, potentially augmented by $d(t)$ padding indices with respect to a policy $B$. If $N^A(t)\geq N^B(t)$, then $\bb^A(t)=\rb^A(t)$ and otherwise, $b^A_i(t)=0$ for $i\leq d(t)$ and $b^A_i(t)=r^A_{i-d(t)}(t)$ for $i\geq d(t)+1$. \\
    $\bb^B(t)$ & Definition symmetrical to $\bb^A(t)$. \\
    $\bbb(t)$ & State descriptor of two M/G/$k$ queues with coupled arrivals, one operating under policy $A$ and the other one under policy $B$. Specifically, $\bbb(t)=(\bb^A(t), \bb^B(t))$. \\
    $d(t)$ & Number of padding indices, that is, the difference in number of jobs at time $t$, between an M/G/$k$ queue under policy $A$ and one under policy $B$. In other words, $d(t)=|N^A(t)-N^B(t)|$.\\
    IND & Integrated number difference, that is, integral of the difference in number of jobs between the SRPT-$k$ and the SEK-SMOD system. \\
    $N^\pi(t)$ & Number of jobs at time $t$ in an M/G/$k$ queue under scheduling policy $\pi$. Note that $N^\pi(t)=|\rb^\pi(t)|$. \\
    $\rb^\pi (t)$ & Vector of ordered remaining sizes at time $t$ in an M/G/$k$ queue under scheduling policy $\pi$. Each element $r_i(t)$ represents the remaining size of the $i$-th smallest job in the queue (including jobs in service). \\
    $r^+(t)$ & Positive part of the difference in work between the SMOD and SRPT-$k$ systems. \\
    $r^-(t)$ & Negative part of the difference in work between the SMOD and SRPT-$k$ systems. \\
    $T^\pi$ & Response time under the scheduling policy $\pi$.  \\
    $W^\pi(t)$ & Total remaining work at time $t$ in an M/G/$k$ queue under scheduling policy $\pi$. Note that $W^\pi(t)=\sum_i r^\pi_i(t)$. \\
    $x<y$ & Parameters of the SEK policy. One of the conditions to trigger SEK is that the largest job has a remaining size between $x$ and $y$. Note that $\epsilon \leq x$. \\
    $\Delta_t$ & Mathematical definition of IND, that is, integral of $N^{SRPT\hy k}(u)-N^{SEK-SMOD}(u)$ for $u$ between times $t$ and $\tau$. \\
    $\Delta_{SRPT\hy k,SMOD}(t,\infty)$ & IND from time $t$ to the future. \\
    $\epsilon'<\epsilon$ & Parameters of SEK policy. One of the conditions to trigger SEK is that $k$ of the $k+1$ jobs present have remaining sizes between $\epsilon'$ and $\epsilon$. \\
    $\tau$ & Starting at time $t$, next time when SRPT-$k$ and SEK-SMOD have identical states. \\
    \hline
  \end{tabular} }
\caption{State descriptor notation.}
\label{tab:notation-states}
\end{table}

\section[Details of Proof of Lemma~\ref{lem:bad}]{Details of proof of Lemma~\ref{lem:bad}}
\label{ap:proof-bad}

To finalize the proof, we need to explicitly compute a lower bound for $\E{\Delta_t \mid\bbb(t) = \bbb \,\&\, \bad_t}$. 

\begin{proof}{Proof of lower bound for $IND$ in bad scenario}
    Let $t_{dom}$ be the first time after $t$ when the SEK-SMOD system dominates the SRPT-$k$ system, or infinity if dominance never occurs.
    \cref{lem:diff-per-job} (combined with \cref{lem:sek-bridge}) shows that $\Delta_t$ is lower bounded by at most $-k \epsilon$ times the number of jobs which arrive prior to $t_{dom}$,
    plus the number of jobs present at the time SMOD begins being used,
    which is at most $k+2$.
    Because $t_{dom}$ is a stopping time with respect to the arrival process,
    by Wald's equation we know that the expected number of jobs which arrive prior to $t_{dom}$ is $\lambda \E{t_{dom}}$.
    Thus, we know that
    \begin{align}
        &\E{\Delta_t \mid\bbb(t) = \bbb \,\&\, \bad_t} 
        \ge - k \epsilon (\lambda \E{t_{dom} \mid\bbb(t) = \bbb \,\&\, \bad_t} + k+2). \label{eq:delta-by-dom}
    \end{align}

    Now, we need to bound $t_{dom}$ for an arbitrary initial state. For this, we invoke \cref{lem:max-diff},
    which (in combination with \cref{lem:sek-bridge}) states that by the time largest job in the SEK divergence scenario completes in the SRPT-$k$ system,
    both systems are merged or SEK-SMOD dominates SRPT-$k$.
    Let $j$ denote this particular job. Thus, $t_{dom}$ is bounded by the remaining response time of this job $j$:
    \begin{align}
        &\E{t_{dom} \mid\bbb(t) = \bbb \,\&\, \bad_t} \le \E{T_j \mid\bbb(t) = \bbb \,\&\, \bad_t}. \label{eq:dom-by-response}
    \end{align}
    
    We now apply the multiserver tagged-job analysis of \cite{grosof-etal-2018-srpt-k}.
    Let the remaining size of the largest job $j$ at the time of divergence $t$ be $r$,
    and let $w$ be the relevant work at time $t$, that is, the total work of remaining size less than or equal to $r$ at time $t$.
    Note that $r \le y$, by the SEK divergence policy,
    and that $w \le k \epsilon + y$,
    as the total relevant work consists only of the $k$ other jobs present at divergence,
    plus the newly arrived job, which can contribute at most $r \le y$ relevant work.
    By the multiserver tagged-job analysis \citep[eq. 5.1 and Lemma 1]{grosof-etal-2018-srpt-k}, we know that
    \begin{align*}
        [T_j \mid\bbb(t) = \bbb \,\&\, \bad_t] &\le_{st} B_{\le y}(kr + w) \\
        \implies \E{T_j \mid\bbb(t) = \bbb \,\&\, \bad_t} &\le \frac{kr + w}{1-\rho_{\le y}} \le \frac{(k+1)y + k\epsilon}{1-\rho_{\le y}},
    \end{align*}
    where $B_{\le y}(\cdot)$ denotes a relevant busy period with relevancy cutoff $y$,
    and $\rho_{\le y} = \lambda \E{S \ind{S \le y}}$ is the relevant load at threshold $y$.

    Note from \cref{def:parameter} that $\epsilon \le x/6$, so we can write:
    \begin{align}
        \label{eq:conditional-response-bound}
        \E{T_j \mid\bbb(t) = \bbb \,\&\, \bad_t} &\le  \frac{(k+1)y + kx/6}{1-\rho_{\le y}}.
    \end{align}
    Combining together \eqref{eq:delta-by-dom}, \eqref{eq:dom-by-response}, and \eqref{eq:conditional-response-bound},
    we find that
    \begin{align*}
        &\E{\Delta_t \mid\bbb(t) = \bbb \,\&\, \bad_t} \ge - k \epsilon \left(\lambda \frac{(k+1)y + kx/6}{1-\rho_{\le y}} + k+2\right).
    \end{align*}
    Because $c_2$ defined in \cref{def:parameter} is the constant multiplying $-\epsilon$ above, the proof is complete. 
    \end{proof}

\section[Details of the proof of Lemma~\ref{lem:good}]{Details of the proof of Lemma~\ref{lem:good}}

\subsection{Probability of a good scenario}
\label{ap:good-probabilities}

We need to show that the probability of a good scenario is lower bounded by $c_3$, defined in \cref{def:parameter}.

\begin{proof}{Proof of probability of good scenario}
Let $b_{k+1}$ be the size of the largest job in divergence starting state $\bbb$. The good scenario specifies that no jobs arrive in the first $k(b_{k+1}-2x/3)$ time,
    then exactly $k$ jobs with sizes in $[x, 2x]$ arrive during the interval
    $[k(b_{k+1}-2x/3), k(b_{k+1}-x/3)]$, and then no jobs arrive until time $k(b_{k+1}+2x)$.

The good scenario consists of three independent events, whose probabilities can be exactly quantified,
    using the fact that the arrival process is a Poisson process with rate $\lambda$:
    \begin{itemize}
        \item No arrivals in the intervals $[0, k(b_{k+1}-2x/3)]$ or $[k(b_{k+1}-x/3), k(b_{k+1}+2x)]$,
        a total duration of $k(b_{k+1}+5x/3)$,
        which has probability $e^{-\lambda k(b_{k+1}+5x/3)} \ge e^{-\lambda k(y+5x/3)}$, because $b_{k+1} \le y$.
        \item Exactly $k$ arrivals during the interval $[k(b_{k+1}-2x/3), k(b_{k+1}-x/3)]$,
        a duration of $kx/3$, which has probability
        $\frac{(\lambda kx/3)^k e^{-\lambda kx/3}}{k!}$.
        \item Each of those $k$ arrivals has size in the range $[x, 2x]$, which has probability
        $(\P{S \in [x, 2x]})^k$. 
        Note that this probability is guaranteed to be positive by the selection of \new{$x$
        in \cref{def:parameter}}.
    \end{itemize}
    The probability that the good event occurs is the product of the probabilities of these independent events, which is lower bounded by the product of the probabilities and lower bounds provided in each bullet point. As desired, \new{this} product equals $c_3$ defined in \cref{def:parameter}. 
    \end{proof}

\subsection{Expected difference in a good scenario}
\label{ap:good-expectation}

To finalize the proof, we prove each of the claims in detail.

\begin{proof}{Proof of claims}
In the rest of the proof, we analyze the effects of each of the following events in the expected $\Delta_t$.

\textbf{The $k-1$ smallest jobs at time of divergence:~~}
    Because no jobs arrive in the first $k(b_{k+1}-2x/3)$ time after divergence,
    and because $b_{k+1} \ge x \ge 6\epsilon$,
    no jobs arrive before the $k$ smallest jobs at time of divergence (in $\bbb$)
    have completed in both systems.
    The only difference between SEK and SRPT-$k$ at the time of divergence is whether the second-largest job is in service (under SRPT-$k$) or the largest job is in service (under SEK).
    Thus, the completion times of all of the other $k-1$ smaller jobs are the same in both systems,
    resulting in no effect on $\Delta_t$.

    \textbf{The second-largest job at time of divergence:~~}
    This job has remaining size $b_k$ at the time of divergence, so it completes at time $kb_k$ under SRPT-$k$ and at $k(b_1+b_k)$ under SEK-SMOD. Then, there is a $-kb_1$ contribution to $\Delta_t$.

    \textbf{The largest job at time of divergence:~~}
    Because additional jobs do not arrive until time $k(b_{k+1}-2x/3)$ at the earliest,
    the largest job at time of divergence has reached a remaining size of no more than
    $2x/3+b_1$ under SRPT-$k$, as it began service at time $k b_1$. Under SEK-SMOD, its remaining size is no more than $2x/3$, as the job began service at the divergence point.
    Because $\new{x} > 6\epsilon$, this job is smaller than all of the arriving jobs, and thus is prioritized
    ahead of the arriving jobs, remaining in service.
    It thus finishes at the same time it would have in the absence of those arrivals,
    finishing $k b_1$ time earlier in the SEK-SMOD system than in the SRPT-$k$ system.
    This gives a $k b_1$ contribution to $\Delta_t$.

    \textbf{The arriving jobs:~~}
    Because of the exclusion of further arrivals under the good scenario, $k-1$ of the arriving jobs enter service immediately and complete without any wait in both systems. 
    The only exception is the job with the largest remaining size at its arrival time,
    which is delayed behind the job with the largest remaining size at divergence,
    the remaining-size-$b_{k+1}$ job.
    This job is the first to complete and free up a server for the largest of the arriving jobs.
    Thus, this largest job's response time
    differs in the two systems by $k$ times the remaining size of this originally-$b_{k+1}$ size job at the time of this arrival.
    The originally-$b_{k+1}$ size job will be in the system at this time, as it completes at time $kb_{k+1}$ at the earliest, after all the arrivals have occurred.

    Throughout this period, the originally-$b_{k+1}$ size job has larger remaining size in the SRPT-$k$ system, by a margin of $b_1$ remaining size.
    Thus, the largest arriving job finishes $k b_1$ later in the SRPT-$k$ system than in the SEK-SMOD system.

    This confirms all of the claims. Because $b_1 \ge \epsilon' = \epsilon/2$,
    we obtain the desired bound with $c_4 = k/2$, as defined in \cref{def:parameter}. 
    \end{proof}

\new{
\section{Positive recurrence of M/G/$k$ scheduling policies}
\label{app:pos-recurrence}
Many results on multiserver scheduling have implicitly assumed ergodicity without an explicit proof of positive recurrence. Here, we give an explicit proof of positive recurrence for the M/G/$k$ under an arbitrary work-conserving scheduling policy. Our proof adapts a proof sketch given in \citet[Appendix D]{hong-et-al2024gittins-ggk} for the more complex G/G/$k$/setup system. We thank Amit Harlev for helpful pointers on establishing this result.

First, we define a work-conserving policy:
\begin{definition}
    A scheduling policy for the M/G/$k$ is work-conserving if it satisfies two properties:
    \begin{itemize}
        \item Whenever at least $k$ jobs are present, $k$ jobs receive service.
        \item Whenever fewer than $k$ jobs are present, all jobs receive service.
    \end{itemize}
\end{definition}

Now, we state and prove the theorem:
\begin{theorem}
    \label{thm:pos-recurrence}
    In the M/G/$k$, for any job size distribution with r.v. $S$, any number of servers $k$, any arrival rate $\lambda > 0$ such that $\rho = \lambda \E{S} < 1$, and any work-conserving scheduling policy $\pi$ with no auxiliary state, the empty state $\bb = \emptyset$ is positive recurrent.
\end{theorem}
\begin{proof}{Proof.}
    We use the Foster-Lyapunov theorem for continuous-time, continuous-state systems \citep{meyn1993survey}, under which a sufficient condition for positive Harris recurrence is the existence of a Lyapunov function $V(\cdot) \ge 0$, a compact set $C$, and constants $\xi > 0,$ and $d$ such that for all states $\bb$ on the M/G/$k$,
    \begin{align*}
        G_\pi \circ V(\bb) \le -\xi + d \ind{\bb \in C},
    \end{align*}
    where $G_\pi$ is the generator for scheduling policy $\pi$.

    To define our Lyapunov function, let $r$ be a large enough constant such that
    \begin{align*}
        \lambda k \E{(S-r)^+} \le (1-\rho)/2.
    \end{align*}
    Such a constant $r$ must exist because $\E{S} < \infty$, which holds because $\rho = \lambda \E{S} < 1$.

    Let $W(\bb)$ denote the total work present in state $\bb$, and let $\max(\bb)$ denote the maximum remaining size of any job in state $\bb$.

    We define our Lyapunov function $V(\bb)$ as follows:
    \begin{align*}
        V(\bb) = W(\bb) + k(\max(\bb) - r)^+.
    \end{align*}
    We now upper bound the drift $G_\pi \circ V(\bb)$.
    
    We start with universal bounds on the drift of the two terms, $W(\bb)$ and $k(\max(\bb) - r)^+$.
    First, in all states, because $W(\bb)$ only increases when jobs arrive, we have $G_\pi \circ W(\bb) \le \lambda \E{S} = \rho$.
    Next, because $k(\max(\bb) - r)^+$ only increases when a job of size at least $r$ arrives, we have 
    \begin{align*}
        G_\pi \circ k(\max(\bb) - r)^+ &\le \lambda k \E{(S-r)^+} \le (1-\rho)/2.
    \end{align*}
    Thus, in all states, $G_\pi \circ V(\bb) \le \rho + (1-\rho)/2 = 1 - (1-\rho)/2$.

    Next, we prove stronger upper bounds in two subsets of the state space: (a) When at least $k$ jobs are present,
    and (b) when fewer than $k$ jobs are present and $\max(\bb) > r$.

    In case (a), when at least $k$ jobs are present,
    by the work conservation property $k$ jobs are in service. As a result, $W(\bb)$ falls at rate 1, outside of arrival moments. As a result,
    \begin{align*}
        G_\pi \circ W(\bb) &= -1 + \lambda \E{S} = -1 + \rho, \\
        G_\pi \circ V(\bb) &\le (-1 + \rho) + (1 -\rho)/2 = -(1-\rho)/2.
    \end{align*}

    Next, if at least one and fewer than $k$ jobs are present, then all jobs are in service, by work conservation, and $\max(\bb)$ falls at rate $1/k$, outside of arrival moments.
    If additionally $\max(\bb) > r$, then $(\max(\bb) - r)^+$ falls at rate $1/k$, outside of arrival moments. As a result,
    \begin{align*}
        G_\pi \circ k(\max(\bb) - r)^+ &\le -1 + \lambda k \E{(S-r)^+} \le -1 + (1-\rho)/2 \\
        G_\pi \circ V(\bb) &\le \rho + (-1 + (1-\rho)/2) = -(1-\rho)/2
    \end{align*}
    
    Define the compact set $C$ to be all states $\bb$ with fewer than $k$ jobs present and $\max(\bb) \le r$.
    Summarizing the above drift bounds, we have shown that
    \begin{align*}
        G_\pi \circ V(\bb) \le -(1-\rho)/2 + \ind{\bb \in C}.
    \end{align*}

    To complete the proof, note that from any state in the set $C$,
    there is a probability $e^{-\lambda k r} > 0$ that there will be no arrivals in the next $k r$ time,
    at which point the system will reach an empty queue.
    Thus, the empty queue state is positive recurrent.
\end{proof}}

\section[Details of Proof of Lemma~\ref{lem:pos-part}]{Details of proof of Lemma~\ref{lem:pos-part}}
\label{ap:pos-part}

To complete the proof, we analyze how arrivals, completions, continuous service and a difference reaching zero influence each job's index and, consequently, the positive difference $r^+$.

\begin{proof}{Proof of effects of events on positive difference}
We analyze the effects on the positive difference of each of the following events:

\textbf{Jump change on arrival:~~}
    We first show that an arrival, occurring in both queues at the same time with the same size,
    does not change the value of $r^+(t)$.
    Specifically, we define the times $t^{(-)}$ and $t^{(+)}$ to denote the left and right limits approaching a given time $t$:
    \begin{align*}
        & r^+(t^{(-)}) := \lim_{\new{u} \to t \mid \new{u} < t} r^+(\new{u}), \quad r^+(t^{(+)}) := \lim_{\new{u} \to t \mid \new{u} > t} r^+(\new{u}).
    \end{align*}
    Our goal is to prove that $r^+(t^{(-)}) = r^+(t^{(+)})$, for any arrival time $t$.

    For concision, we write the padded vectors $b^{SMOD}(t^{(-)})$ and $b^{SRPT\hy k}(t^{(-)})$, just before the arrival at time $t$, as $p$ and $q$ respectively.
    Specifically, let the remaining sizes of the SMOD system just prior to the arrival be $p_1 \le  p_2 \le  p_3 \le \ldots$
    and let the remaining sizes in the SRPT-$k$ system just prior to the arrival be $q_1 \le q_2 \le q_3 \le \ldots$. Note that some of these values may be zero, due to the padding.

    Let the arriving job have size $\new{s}$, and let its index upon arrival in the SMOD system be $i$ and in the SRPT-$k$ system be $j$. In other words, $p_{i-1} \le \new{s} \le p_i$, and $q_{j-1} \le \new{s} \le q_j$.

    Note that all indices $\ell$ less than $\min(i-1, j-1)$ or more than $\max(i, j)$ preserve the difference $p_\ell-q_\ell$ before and after the arrival -- for small indices, the arrival is irrelevant, and for large indices, each job's index is increased by 1, leaving the difference unchanged. Thus, we only need to look at the effect on $r^+$ from indices between $i$ and $j$ (inclusive).

    \new{We} consider three cases: $i=j$, $i<j$, $i>j$. 
    \begin{itemize}
        \item If $i=j$, the arriving job has the same index in both systems, contributing 0 to $r^+$. 
    
        \item If $i<j$, the arriving job of size $\new{s}$ affects the indices between $i$ and $j$ as follows:
    \begin{align*}
        &\begin{array}{c} 
            b^{SMOD}(t^{(-)}): \\ b^{SRPT\hy k}(t^{(-)}): \\
        \end{array}
        \begin{array}{cccc}
            p_i, & p_{i+1}, & \cdots, & p_{j-1}; \\
            q_i, & q_{i+1}, & \cdots, & q_{j-1}; \\
        \end{array} 
        \begin{array}{c} 
            b^{SMOD}(t^{(+)}): \\ b^{SRPT\hy k}(t^{(+)}): \\
        \end{array}
        \begin{array}{lllll}
            \new{s}, & p_i, & \cdots, & p_{j-2}, & p_{j-1} \\
            q_i, & q_{i+1}, & \cdots, & q_{j-1}, & \new{s} \\
        \end{array} .
    \end{align*}
    By definition, we have
    \begin{align*}
        r^+(t^{(-)}) &= \sum_{\ell=i}^{j-1} (p_\ell - q_\ell)^+, \; \text{and } 
        r^+(t^{(+)}) = (\new{s} - q_i)^+ + (p_{j-1} - \new{s})^+ 
        + \sum_{\ell=i}^{j-2} (p_\ell - q_{\ell+1})^+ .
    \end{align*}
  
      Note that, for all indices $i'>i$ and $j'<j$ we have $q_{j'} \le \new{s} \le p_{i'}$ because our state descriptor orders the remaining sizes from smallest to largest.
    Thus, for each index $\ell \in (i, j)$, the difference before arrival $(p_\ell - q_\ell)^+ = p_\ell - q_\ell$ is nonnegative. Likewise, for each index the difference after arrival $(p_\ell - q_{\ell+1})^+ = p_\ell - q_{\ell+1}$ is also nonnegative. Finally, the difference $(\new{s} - q_i)^+ = \new{s}-q_i$ and $(p_{j-1} - \new{s})^+ = p_{j-1} - \new{s}$. Thus, we find that $r^+(t^{(-)})=r^+(t^{(+)})$.

        \item Finally, if $i>j$, we have the reversed \new{situation}. Following a similar argument, $r^+(t^{(-)})=r^+(t^{(+)})$ as well.

    \end{itemize}

    \textbf{Jump change on completions:~~} 
    Similarly to the arrivals, we split into three cases: A completion where both systems have an equal number of jobs, a completion in the system with fewer jobs, and in the system with more jobs. 
    \begin{itemize}
        \item In the first case, supposing that the completion happens at time $t$, we have the following \new{situation}, where $\delta$ represents the job very close to completion:
    \begin{align*}
        &\begin{array}{c} 
            b^{SMOD}(t^{(-)}): \\ b^{SRPT\hy k}(t^{(-)}): \\
        \end{array}
        \begin{array}{cccc}
            \delta, & p_2, & p_3, & \cdots; \\
            q_1, & q_2, & q_3, & \cdots; \\
        \end{array} 
        \begin{array}{c} 
            b^{SMOD}(t^{(+)}): \\ b^{SRPT\hy k}(t^{(+)}): \\
        \end{array}
        \begin{array}{cccc}
            0, & p_2, & p_3, & \cdots \\
            q_1, & q_2, & q_3, & \cdots \\
        \end{array} .
    \end{align*}
    Note that $q_1 > 0$, because padding only occurs to a system with \new{fewer} jobs after a completion. The $0$ appears in the SMOD system due to that padding. In this step, we ignore the change in the remaining sizes of other jobs, as we handle this case separately in the proof. Here we focus on the effect of the completion. 
    
    The state space changes continuously in each index, so $r^+$ \new{does not} change as a result of this kind of completion. We showed the effect of a completion in the SMOD system, but note that a completion in the SRPT-$k$ system is equivalent.

    \item If the completion happens in the system with fewer jobs, assuming this is the SMOD system, we have
        \begin{align*}
        &\begin{array}{c} 
            b^{SMOD}(t^{(-)}): \\ b^{SRPT\hy k}(t^{(-)}): \\
        \end{array}
        \begin{array}{ccccc}
            \cdots & 0, & \delta, & p_{i+2}, & \cdots; \\
            \cdots & q_i, & q_{i+1}, & q_{i+2}, & \cdots; \\
        \end{array} 
        \begin{array}{c} 
            b^{SMOD}(t^{(+)}): \\ b^{SRPT\hy k}(t^{(+)}): \\
        \end{array}
        \begin{array}{cccccc}
            \cdots & 0, & 0, & p_{i+2}, & \cdots \\
            \cdots & q_i, & q_{i+1}, & q_{i+2}, & \cdots \\
        \end{array} .
    \end{align*}

    We again see a 0 appear where the near-completion job was, causing no change to any indices or $r^+$. Again, if SMOD and SRPT-$k$ are swapped, nothing changes.

    \item Finally, if a completion occurs in the system with more jobs, 
    the only change as a result of the completion is the removal of the first index. This index induces no difference at the time of the completion. All other pairs of jobs remain paired and do not change their differences. Thus, $r^+$ does not change as a result of the completion.
    \end{itemize}

    \textbf{Continuous change due to work being completed:~~} We again consider three cases: equal number of jobs in both systems, more jobs in the SRPT-$k$ system, and more jobs in the SMOD system.
    \begin{itemize}
        \item $N^{SRPT\hy k}=N^{SMOD}$: If both systems have the same number of jobs, then both policies will serve the same set of indices: the $k$ jobs of least remaining size, or all of the jobs if there \new{are fewer} than $k$. No differences will change, and the positive part $r^+$ is constant.

        \item $N^{SRPT\hy k}>N^{SMOD}$: If there are more jobs in the SRPT-$k$ system, both policies serve the $k$ jobs of least remaining size. In particular, every index being served in the SRPT-$k$ system corresponds to an index that is being either served in the SMOD system or is an empty padding index. In the former case, the difference does not change. In the latter case, the difference is negative, and has zero positive part. The only way $r^+$ can change is due to the other indices that are served in the SMOD system but not in the SRPT-$k$ system, and this decreases $r^+$.

        \item $N^{SRPT\hy k}<N^{SMOD}$: In this case, we closely look at the \textit{eligible jobs} for the SMOD policy. From \cref{def:smod}, recall that the $k+d$ jobs with the smallest size are eligible for service (where $d=N^{SMOD}-N^{SRPT\hy k}$). Among these $k+d$ jobs, SMOD gives high priority to jobs with a remaining size of at least the corresponding SRPT-$k$ job, which we referred to as \textit{zero-or-positive difference indices}. Therefore, the set of \textit{eligible jobs} contains all jobs at indices that are either being served in the SRPT-$k$ system, or are empty padding indices.
        
        Let us condition on the number of indices of eligible jobs that have zero-or-positive-difference $b_i^{SMOD} - b_i^{SRPT\hy k}$. The rate of change of $r^+$ can be characterized based on the number of such indices being served in each system:
        \begin{itemize}
            \item For each index with zero-or-positive-difference that both SMOD and SRPT-$k$ serve, $r^+$ is preserved since both remaining sizes decrease at a rate of $1/k$. 
            \item For each index with zero-or-positive index that SRPT-$k$ serves but SMOD does not, $r^+$ increases at a rate of $1/k$. 
            \item For each index that SMOD serves but SRPT-$k$ does not, $r^+$ decreases at a rate of $1/k$. By $\cref{def:smod}$, such an index must be a padding index, and so has a positive difference.
            \item All other services do not affect $r^+$: Service to negative-difference indices, or service by SMOD to zero-difference indices which are not being served by SRPT-$k$.
        \end{itemize}

        Note that SMOD avoids ever serving indices in this last group (zero-difference indices served by SMOD but not SRPT-$k$). From \cref{def:smod}, SMOD only serves indices that are either padding indices (which always have positive difference) or that are also being served in the SRPT-$k$ system. This is a key property of the SMOD policy.
        
        Thus, to prove that $r^+$ never increases, it suffices to show that the number of zero-or-positive-difference indices served by SMOD is at least the number of nonnegative indices served by SRPT-$k$.
        
        If there are at least $k$ nonnegative indices within the eligible set, SMOD serves exactly $k$ such indices, as it prioritizes zero-or-positive difference indices over all negative-difference indices.
        This corresponds to at least as many zero-or-positive difference indices as SRPT-$k$ because only $k$ jobs can be served simultaneously.
        
        If there are fewer than $k$ such indices within the eligible set, SMOD serves all of them. Because SRPT-$k$ only serves jobs that are within the eligible set, SMOD again serves at least as many zero-or-positive difference indices as SRPT-$k$.
    \end{itemize}

    This completes the proof. 
    \end{proof}

\section[Details of Proof of Lemma~\ref{lem:zig-zag}]{Details of proof of Lemma~\ref{lem:zig-zag}}
\label{ap:zig-zag-proof}

To finalize the proof, we show that zig-zag matching is preserved under arrivals, completions, and continuous service.

\begin{proof}{Proof of zig-zag \new{preservation}}
We analyze the effect of each of the three events in the zig-zag matching property.

\textbf{Jump change on arrivals:~~} Let $\new{s}$ be the size of the arriving job, $i$ the index at which it is inserted into the SMOD system's $r$ vector at time $t$, and $t^{(-)}$ the left limit as we approach time $t$,
    as in the proof of \cref{lem:pos-part}.
    
    We have $\new{s} \le b^{SMOD}_i(t^{(-)})$, as $\new{s}$ is inserted between $b^{SMOD}_{i-1}(t^{(-)})$ and $b^{SMOD}_i(t^{(-)})$.
    By the zig-zag property, we thus have $\new{s} \le b^{SRPT\hy k}_{i+1}(t^{(-)})$.
    Thus, $\new{s}$ is inserted at index at most $i+1$ in the SRPT-$k$ system,
    at most 1 position after its location in the SMOD system.

    \new{We} split into two cases: When $\new{s}$ is inserted at index $i+1$ in the SRPT-$k$ system,
    and when it is inserted at an index $\le i$.
    \begin{itemize}
        \item In the first case, the two insertion locations are matched with each other, and no matchings change, so zig-zag matching is preserved.

        \item In the second case, let $j$ be the insertion location in the SRPT-$k$ system. Notice that all jobs in the SMOD system up to index $i$ have remaining size at most $\new{s}$, while all jobs in the SRPT-$k$ system at or above index $j$ have value at least $\new{s}$. Thus all matchings from $b^{SMOD}_{j-1}$ to $b^{SMOD}_i$ obey the zig-zag-matching property. Note also that no other matchings change.
    \end{itemize}
    
    Thus, arrivals maintain the zig-zag property.

    \textbf{Jump change in service:~~} If a completion occurs in the system with equal or more jobs, the index will become a padding index. Because the systems change continuously, the property is preserved.
    Otherwise, there is a padding index present in the other system, resulting in a pair of zeros being removed, which does not affect the zig-zag matching property.

    \textbf{Continuous change in service:~~} There are three cases: Equal jobs in both systems, fewer jobs in the SMOD system, or more jobs in the SMOD system. \new{We analyze each case separately:}
    \begin{itemize}
        \item If there are equal numbers of jobs in both systems, then both systems serve the first $k$ indices. This decreases $b_i^{SMOD}$ and $b_i^{SRPT\hy k}$ at rate $1/k$ for all $i \le k$. In particular, this preserves the zig-zag matching -- each zig-zag difference $b_{i+1}^{SRPT\hy k} - b_i^{SMOD}$ is unchanged for $i < k$, and can only increase for $i = k$.

        \item If there are fewer jobs in the SMOD system, then every job served by SRPT-$k$ is either zig-zag matched with a job that is served by SMOD, or with a padding index. In particular, SMOD is serving at least up to index $k+1$, while SRPT-$k$ is serving at most index $k$. Thus, the inequality is preserved.

        \item If there are more jobs in the SMOD system, there can only be one additional job, as proven above. Given that SMOD has exactly one additional job, the SMOD eligible set consists of the $k+1$ jobs of least remaining size in the SMOD system, and SMOD serves all but 1 such job. 
        
        There are two cases to consider: All of the indices in the eligible set have nonnegative difference, or there is a negative difference index.
        \begin{itemize}
            \item In the former case, SMOD will serve the $k$ jobs of least remaining size, as that is its tiebreaker. In particular, SMOD serves indices $1$ through $k$, while SRPT-$k$ serves indices $2$ through $k+1$ due to its one padding index. All service is zig-zag matched, and no inequalities are violated.

            \item In the latter case, SMOD will omit a job $\new{i}$ with a negative-diff index from service. We have $b^{SMOD}_i < b^{SRPT\hy k}_i$, because of the negative difference, and $b^{SRPT\hy k}_i \le b^{SRPT\hy k}_{i+1}$, due to the ordering of the indices. Thus, we have $b^{SMOD}_i < b^{SRPT\hy k}_{i+1}$. In other words, the zig-zag inequality is strict in the index that SMOD omits from service, and by the time it would become nonstrict (an equality), SMOD will change its service.

            For all other indices that are zig-zag matched with an SRPT-$k$ index (all indices $\le k$), SMOD serves those indices, maintaining the inequalities.
        \end{itemize}
    Thus, zig-zag matching is preserved by the continuous change of service.
    \end{itemize}

This completes the proof. 
\end{proof}

\section[Details of Proof of Lemma~\ref{lem:pln}]{Details of Proof of Lemma~\ref{lem:pln}}
\label{ap:proof-pln}

To finalize the proof, we only need to show that PLN is preserved. 

\begin{proof}{Proof of PLN preserved}
We analyze the effect on the PLN property of each of the following events in the system: arrival, continuous service, completion, and a difference reaching zero.

    \textbf{Arrival:~~} Let the arriving job have size $\new{s}$, $i$ be the index where the job is inserted in the SMOD system, and $j$ the corresponding index in the SRPT-$k$ system.
    Let the arrival be at time $t$. For ease of notation (similarly to \cref{ap:pos-part}), let $p$ denote the SMOD vector $b^{SMOD}(t^{(-)})$ just before time $t$, and $q$ the SRPT-$k$ vector $b^{SRPT\hy k}(t^{(-)})$. We study the cases $i=j$, $i<j$ and $i>j$ separately.
    \begin{itemize}
        \item If $i = j$, the new job is placed at the same index in both systems. Then, the sequence of differences only changes by the insertion of a new zero-diff entry at index $i$. In particular, the relative ordering of positive and negative diff indices is unchanged, preserving PLN.

        \item If $i < j$, we see the following changes in $p$ and $q$ at time $t$:
        \begin{align*}
            &\begin{array}{c}
                b^{SMOD}(t^{(-)}): \\ b^{SRPT\hy k}(t^{(-)}):
            \end{array}
            \begin{array}{llll}
                p_i, & p_{i+1}, & \cdots & p_{j-1}; \\
                q_i, & q_{i+1}, & \cdots & q_{j-1}; \\
            \end{array}
            &\begin{array}{c}
                b^{SMOD}(t^{(+)}): \\ b^{SRPT\hy k}(t^{(+)}):
            \end{array}
            \begin{array}{lllll}
                \new{s}, & p_i, & \cdots & p_{j-2}, & p_{j-1} \\
                q_i, & q_{i+1}, & \cdots & q_{j-1}, & \new{s} \\
            \end{array}
        \end{align*}
    
    Note that the only diffs that are affected are indices in the interval $[i, j]$ at time $t^{(+)}$.
    We thus focus on PLN relationships involving these indices. For all $\ell \ge i$, we have $p_\ell \ge \new{s}$,
    and for all $m < j$, $q_m \le \new{s}$.
    Thus, all indices in the interval $[i, j-1]$ were zero-or-positive-diff indices at time $t^{(-)}$.
    
    If all such indices were zero-diff, then all of the indices involved must have remaining size exactly $\new{s}$:
    $p_i = p_{i+1} = \ldots = p_{j-1} = q_i = \ldots = q_{j-1} = \new{s}$.
    In this case, all indices in the interval $[i, j]$ are zero-diff at time $t^{(+)}$, maintaining PLN regardless of which other indices are positive or negative.
    
    Now, we turn to the case where at least one index in the interval $[i, j-1]$ was \new{positive}-diff at time $t^{(-)}$.
    Because the PLN property holds at time $t^{(-)}$, all positive-diff indices at time $t^{(-)}$ are greater than $j$.
    Then, only zero-diff and negative-diff indices are created in the interval $[i, j]$ at time $t^{(+)}$, preserving PLN.

    \item If $i > j$, \new{the argument is similar to the case $i<j$. }
    Either all indices in the interval $\new{[j, i-1]}$ at time $t^{(-)}$ were zero-diff and equal to \new{$s$}, or there was a positive-diff index. In the first case, PLN is preserved regardless of which other indices are positive or negative. In the latter case, because PLN held at time $t^{(-)}$, any positive-diff indices at time $t^{(-)}$ are below $j$. Then, they create new positive-diff indices at time $t^{(+)}$ in the interval \new{$[j,i]$} and PLN is preserved.
    \end{itemize}
    
    \textbf{Continuous service:~~} If $N^{SMOD} = N^{SRPT\hy k}$, then both systems serve the same set of indices, and the difference $b^{SRPT\hy k}_i - b^{SMOD}_i$ is constant for all indices $i$. Hence, PLN is maintained.

    If $N^{SMOD} = N^{SRPT\hy k} + 1$, we must examine the first $k+1$ indices to determine SMOD's service.
    Within this case, if $N^{SRPT\hy k} < k$, all indices are served in the SMOD system. The only index for which the diff $b^{SRPT\hy k}_i - b^{SMOD}_i$ changes is the index $i=1$.
    Index $i=1$ is a positive diff index, because $b^{SRPT\hy k}_1 > 0$ and $b^{SMOD}_1 = 0$.
    As service occurs, the differences \new{decrease} towards 0, which does not affect PLN.
    No other diffs change, so PLN is maintained.
    
    Otherwise, within the case $N^{SMOD} = N^{SRPT\new{\hy k}} + 1$, there are two possibilities, depending on whether there is a negative-diff index among the first $k+1$ indices.
    \begin{itemize}
        \item If there is no negative diff among the first $k+1$ indices, SMOD serves the first $k$ indices.
        SRPT-$k$, instead, serves the $k+1$ index, so the diff of index $k+1$ increases at rate $1/k$. If index $k+1$ previously had zero diff, this may create a new positive-diff index. 
        
        Regardless of whether a new positive-diff index is created, this service maintains the PLN property because all negative-diff indices are larger than the $k+1$ index, by the assumption of this case. 
        
        The only other index where service differs between the two systems is index 1, where SMOD has a padding zero. This is a positive-diff index, so that service also maintains PLN. All other indices receive the same service in both systems, maintaining the diffs and PLN.
        
        \item If there is a negative diff among the first $k+1$ indices, SMOD omits the largest negative-diff index from service, and serves the rest of the first $k+1$ indices. 
        
        SMOD's strategy to omit one such index from service causes the negative-diff of that index to increase towards 0. This maintains PLN, as only creating new nonzero diffs can affect PLN. All other indices maintain PLN for the same reasons as in the first case.  Note that when the diff reaches zero, the scheduling policy changes. This corresponds to a completion, which we handle later in the proof.
    \end{itemize}

    \textbf{Completions:~~}
    Completions do not change the ordering of positive, zero, or negative differences.
    In particular, a completion may lead to a padding index being created or removed.
    Creating a padding index when a job completes has no effect on the differences,
    as the remaining size of the index continuously changes to 0.
    Removing a padding index shifts all remaining indices down by one when a job completes in the system with fewer jobs present. However, it does not change the order of the differences.
    In particular, it always maintains the PLN property.

    \textbf{A difference reaches zero:~~}
    At the moment that indices reach zero difference, the SMOD service option changes discontinuously. 
    However, the PLN property is always maintained if a positive or negative-diff index
    becomes a zero-diff index, so these discontinuities maintain the PLN property.

    Thus, in all cases, PLN is preserved, as desired. 
    \end{proof}

\section[Details of Proof of Lemma~\ref{lem:max-diff}]{Details of proof of Lemma~\ref{lem:max-diff}}
\label{ap:proof-max-diff}

Similarly to the approach in the previous proofs, we analyze in detail the effect of each event.

\begin{proof}{Proof of effects of events on PLN}
We analyze the effects of the following events: arrival, continuous service, completion, and a difference reaching zero.

\textbf{Arrival:~~} As discussed in \cref{lem:pln}, arrivals can only create nonzero-diff indices if the size of the arriving job is no more than the size of a job at a nonzero-diff index in one of the two systems.
    As $i_*(t_2)$ is the largest nonzero-diff index, and it is a negative-diff index where $b^{SRPT\hy k}_{i_*(t_2)} > b^{SMOD}_{i_*(t_2)}$, the arriving job must have size smaller than
    $b^{SRPT\hy k}_{i_*(t_2)}$. As a result, the arriving job is inserted at an index below both $i_*(t_2^{(-)})$ and $m_{t_1}(t_2^{(-)})$, increasing both indices by one
    and maintaining the property $i_*(t_2) \le m_{t_1}(t_2)$.

    \textbf{Continuous service:~~} We follow the same approach as in the proof of \cref{lem:pln} and split into two cases: $N^{SMOD}(t_2) = N^{SRPT\hy k}(t_2)$ or $N^{SMOD}(t_2) = N^{SRPT\hy k}(t_2) + 1$.
    \begin{itemize}
        \item If $N^{SMOD}(t_2) = N^{SRPT\hy k}(t_2)$, the same indices are served in both systems, so no differences change. Hence, $i_*(t_2)$ and $m_{t_1}(t_2)$ do not change, so the property is maintained.

        \item If $N^{SMOD}(t_2) = N^{SRPT\hy k}(t_2) + 1$,
    we must verify that $i_*(t_2)$ does not increase. Note that no new index with nonzero difference is created.
    If $i_*(t_2)>k+1$, then no service occurs to indices at or above $i_*(t_2)$ and their differences do not change. Otherwise, the $i_*(t_2)$ index is served in the SRPT-$k$ system.
    In this case, the $i_*(t_2)$ index is omitted from service by SMOD, because it must be the largest negative-diff index among those $k+1$ indices that are eligible for service.
    Thus, all indices above $i_*(t_2)$ up to $k+1$ are served by SMOD.
    As these indices are served in both systems, their differences do not change and $i_*(t_2)$ does not increase, as desired. Hence, the property is preserved.
    \end{itemize}

    \textbf{Completion:~~} A completion affects the system similarly: a new nonzero-diff index cannot be created, so if a padding index is removed in both systems, both $i_*(t_2^-)$ and $m_{t_1}(t_2^-)$ decrease by 1, and otherwise $i_*(t_2^-)$ and $m_{t_1}(t_2^-)$ are unchanged. In either case, $i_*(t_2) \le m_{t_1}(t_2)$ is maintained.

    \textbf{A difference reaches zero:~~} 
    When a (previously) nonzero difference reaches zero, $i_*(t_2)$ may decrease. If $i_*(t_2^-)$ was the only nonzero-diff index, then the system states have become identical and SMOD dominates SRPT-$k$,
    which persists until $t_{comp}$.
    But $i_*(t_2)$ cannot increase, because no new nonzero-diff index was created.
    Thus, $i_*(t_2) \le m_{t_1}(t_2)$ is maintained.

    We find that in all cases, the $i_*(t_2) \le m_{t_1}(t_2)$ property is maintained until $t_{comp}$ or until SMOD dominates SRPT-$k$, so at time $t_{comp}$ SMOD dominates SRPT-$k$. 
    \end{proof}



\end{document}